\documentclass[10pt,twocolumn,journal]{IEEEtran}
\usepackage{latexsym}
\usepackage{amsfonts}
\usepackage{amsbsy}
\usepackage{amsmath,amssymb}
\usepackage{amsthm}
\usepackage{times}
\usepackage{graphicx}
\usepackage{enumerate}
\usepackage[usenames]{color}
\usepackage[dvips]{pstcol}
\usepackage{epstopdf}
\usepackage{cite}
\usepackage{amsmath}
\usepackage{amssymb}
\usepackage{amsfonts}
\usepackage{graphicx}
\usepackage{epsfig}
\usepackage{psfrag}
\usepackage{xcolor}
\usepackage{amsfonts, bm}
\usepackage{epstopdf}
\usepackage{cite}
\usepackage{color}
\usepackage{xcolor}
\usepackage{subfig}
\usepackage{algorithm}
\usepackage{amsmath}
\usepackage{lipsum}
\usepackage{stfloats}
\usepackage{makecell}

\newtheorem{theorem}{Theorem}
\newtheorem{lemma}{Lemma}
\newtheorem{assumption}{Assumption}
\newtheorem{remark}{Remark}
\newtheorem{corollary}{Corollary}
\newtheorem{proposition}{Proposition}

\newcommand{\tabincell}[2]{\begin{tabular}{@{}#1@{}}#2\end{tabular}}

\input epsf
\begin{document}
	
	\title{ Non-Orthogonal Multiple Access for UAV-Aided Heterogeneous Networks: A Stochastic Geometry Model}
	
	\author{Cunzhuo Zhao, Yuanwei Liu, \IEEEmembership{Senior Member,~IEEE,} Yunlong Cai, \IEEEmembership{Senior Member,~IEEE,}  Minjian Zhao, \IEEEmembership{Member,~IEEE,} and Zhiguo Ding, \IEEEmembership{Fellow,~IEEE}
		
		\thanks{
			C. Zhao is with Pengcheng Laboratory, Shenzhen 518052, China (Email: zhaoczh@pcl.ac.cn)
			
			Y. Liu is with the School of Electronic Engineering and Computer Science, Queen Mary University of London, London E1 4NS, U.K. (Email: yuanwei.liu@qmul.ac.uk).
			
			Y. Cai and M. Zhao are with the College of Information Science and Electronic Engineering, Zhejiang University, Hangzhou 310027, China (Email: ylcai521@gmail.com; mjzhao@zju.edu.cn).

			Z. Ding is with the School of Electrical and Electronic Engineering, The University of Manchester, Manchester M13 9PL, U.K. (Email: zhiguo.ding@manchester.ac.uk).		
		}
		
	}
	\maketitle
	\vspace{-3.3em}
	
	\begin{abstract}
		In this work, we explore the potential benefits of deploying unmanned aerial vehicles (UAVs) as aerial base stations (ABSs) with sub-6GHz band and small cells terrestrial base stations (TBSs) with millimeter wave (mmWave) band in a hybrid heterogeneous networks (HetNets).
		A flexible non-orthogonal multiple access (NOMA) based user  association policy is proposed.  By using the tools from stochastic geometry, new  analytical expressions for association  probability, coverage probability and spectrum efficiency are derived for  characterizing the performance of UAV-aided HetNets under the realistic Air-to-Ground (A2G) channels and the Ground-to-Ground (G2G) channels with a LoS ball blockage model. Finally, we provide insights on the proposed hybrid  HetNets by numerical results. We confirm that i) the proposed NOMA enabled HetNets is  capable of achieving superior performance compared with the  OMA enabled ABSs by setting power allocation factors and targeted signal-to-interference-plus-noise ratio (SINR) threshold properly;  ii) there is a tradeoff between the association probabilities and  the spectrum efficiency in the NOMA enabled ABSs tier; iii) the coverage probability and spectrum efficiency of the NOMA enabled ABSs tier is largely affected by the imperfect successive interference cancellation (ipSIC) coefficient, power allocation factors and SINR threshold; iv)  compared with only sub-6GHz  ABSs, mmWave enabled TBSs are capable of enhancing the spectrum efficiency  of the HetNets when the mmWave line-of-sight (LoS) link is available.
	\end{abstract}
	\begin{IEEEkeywords}
		HetNets, mmWave, NOMA,  stochastic geometry, UAV.
	\end{IEEEkeywords}

	\IEEEpeerreviewmaketitle
	
	\section{Introduction}

	In recent years, unmanned aerial vehicles (UAVs), commonly known as drones, have become a hotspot  for wireless communications due to its unique attributes such as low-cost, mobility, and flexible reconfiguration\cite{UAV_1}. In the meantime, in the process of standardization for 5G/B5G networks, UAVs are gradually being considered as a critical candidate to support diverse applications, such as reconnaissance, remote sensing, or working as temporal base stations\cite{UAV_2}. The popularity of UAVs motivates the researchers to explore the opportunities for integrating UAVs into the existing wireless networks.
	
	In the case of unexpected and temporary events, which creating hard-to-predict inhomogeneous traffic demand, such as  traffic congestions, or event sports, it is difficult to achieve universal connectivity due to the huge number of users together with the severe path loss and excessive inter-cell interference \cite{new_2}. One efficient approach to improve the coverage in currently deployed terrestrial cellular networks  is to equip the UAVs as aerial base stations (ABS), augmented with the functionalities of terrestrial base stations (TBSs)\cite{UAV_3}. As compared to  terrestrial cellular networks, one distinct feature of UAV communications is that the Air-to-Ground (A2G) links are more likely to experience line-of-sight (LoS) propagation which offers lower attenuation\cite{ABS_channel}. To further exploit the spectrum efficiency of the UAV-aided networks, non-orthogonal multiple access (NOMA) has attracted  much attention for its capability of serving multiple user equipments (UEs) at different quality-of-service (QoS) requirements in the same resource block\cite{UAV_NOMA1,UAV_NOMA2}. The key idea of NOMA is to employ a superposition coding (SC) at the transmitter and successive interference cancellation  (SIC) at the receiver\cite{NOMA_def}, which provides a good trade-off between the throughput of the system and the UEs fairness. Therefore, by adopting NOMA techniques, the achievable spectrum efficiency of the networks can be improved.

	On the other hand,  using high-frequency band and densification  will be two key capacity-increasing techniques for cellular networks, such as millimeter wave (mmWave) communications\cite{mmWave_def1} and small cells. Deploying  terrestrial mmWave small cells will offer high capacity when a connection is available. Motivated by the benefits of UAV-aided networks, NOMA techniques, and mmWave transmissions,  in this work we consider a UAV-aided heterogeneous network (HetNets) where a mmWave terrestrial network  co-exists with a sub-6GHz NOMA enabled aerial network. Note that  the use of mmWave and microWave resources simultaneously is a feature of 5G/B5G networks, and their distinctive carrier frequencies avoid the inter-tier interference\cite{mmWave_def2}.

	\subsection{Related Work and Motivation}
	Modeling and analyzing cellular networks  with the aid of stochastic geometry has been widely adopted due to its accuracy and tractability. In the  studies of UAV-aided networks, the authors of \cite{HetNet_uav1} derived the coverage probability for a finite ABS network by modeling the locations of ABSs as a uniform binomial point process (BPP). The authors \cite{HetNet_uav3}  analyzed the  downlink coverage performance of UAV-aided cellular networks when the UEs are clustered around the projections of ABSs on the ground. A framework was  proposed in \cite{HetNet_uav2} to analyze the behaviors of a ABSs network under a realistic A2G channel model which incorporates the LoS and  non-line-of-sight  (NLoS) links.  This work was further extended in  \cite{HetNet_ref} where the network comprises both ABSs and TBSs. Besides, instead of considering the average probabilistic path loss in most of works, the authors of \cite{HetNet_uav2,HetNet_ref}  considered more realistic LoS and NLoS transmissions, respectively. Multi-tier UAV-aided networks were presented in \cite{UAV_3,HetNet_uav5,HetNet_uav6}.  Specifically, the authors of \cite{UAV_3} and \cite{HetNet_uav5} proposed the multi-tier drone architecture based on the standard terrestrial  path loss model.  Furthermore,  the multi-tier UAV-aided networks based on the transmitter-oriented or receiver-oriented rules under a realistic A2G channel model were studied in \cite{HetNet_uav6}.
	In \cite{new_4},  a network comprising ABSs and TBSs whose location follow BPP and PPP is analyzed. The authors of \cite{HetNet_uav7} advocated a pair of strategies in UAV-aided NOMA  networks, i.e., the UAV-centric strategy for offloading actions and the user-centric strategy  for providing emergency communications. In \cite{HetNet_uav8}, a  multiple-input multiple-output (MIMO)-NOMA enabled UAV network was proposed, where the outage probability and ergodic rate were evaluated in the downlink scenario.

	Regarding the literature of stochastic geometry based HetNets systems, a  system containing sub-6GHz macrocells and mmWave small cells was analyzed in \cite{HetNet_mmWave1}, where the  macrocells provide universal coverage and the small cells provide high data rate when the  mmWave LoS link is available. The authors of \cite{HetNet_mmWave2} studied the decoupled association in a sub-6GHz and mmWave deployment  from the resource allocation perspective. A meta distributions of the SIR/SNR and data rate of a typical device in a hybrid spectrum network and $\mu$wave-only network is characterized in \cite{new_5}. Building upon the above research contributions and the analytical tools of stochastic geometry, we propose an architecture of UAV-aided HetNets where mmWave terrestrial networks co-exist with a sub-6GHz aerial networks, which has not been well studied in the literature. In contrast to the previously reported designs of UAV-aided HetNets \cite{HetNet_ref,HetNet_uav6}, our proposed  architecture  poses three additional challenges: i) The NOMA techniques causes additional interference from the connected ABS to the served UE; ii) The channel ordering  needs to be determined  under the unique  characteristics of A2G channels; iii) The UE association policy needs to be carefully designed under the existence of mmWave TBSs and sub-6GHz ABSs.

	\subsection{Contributions and Organization}
	The primary contributions of this paper can be summarized as follows:
	\begin{itemize}
		\item 	
		By taking advantage of unique attributes of UAVs and high transmission rate of mmWave, we propose a new  model of  HetNets where the sub-6GHz NOMA enabled ABSs overlay the mmWave TBSs. We model the Ground-to-Ground (G2G) and A2G links incorporating the impact of LoS and NLoS path loss attenuations. We consider the LoS and NLoS transmissions separately, where  two independent non-homogeneous poisson point processes (PPPs)  are formulated.
		\item
		We develop a flexible association policy to address the co-existence of NOMA enabled ABSs and mmWave TBSs. Under this policy,  we first derive the analytical expressions for the distance distributions given that the typical UE is associated
		with a TBS, a NLoS ABS or a LoS ABS.
		\item
		We derive exact analytical expressions for the UAV-aided HetNets in terms of coverage probability and spectrum efficiency. Additionally, the effect of the imperfect successive interference cancellation (ipSIC)  is analyzed.
		\item
		We provide the basic  power allocation guidelines for the NOMA enabled networks, in which the targeted signal-to-interference-plus-noise ratio (SINR) threshold of the typical UE and the fixed UE both determine the coverage probability of a typical UE. We also provide insights on the HetNets design by numerical results,  which demonstrate that our proposed NOMA enabled HetNets is capable of achieving superior performance compared with the conventional OMA enabled HetNets.

	\end{itemize}

	The rest of this paper is organized as follows. In Section II, the HetNets model and the association strategy are introduced. In Section III, new analytical expressions for distance distributions and association probabilities are derived. Then the coverage probability and spectrum efficiency of the network are investigated in Section IV and Section V, respectively. Our numerical results are demonstrated in Section VI, which is followed by the conclusions in Section VII.
	Here we give a flow chart  with analytical steps shown in Fig. \ref{flow_chart}.

	\begin{figure*}[!t]
		\centering
		\scalebox{0.48}{\includegraphics{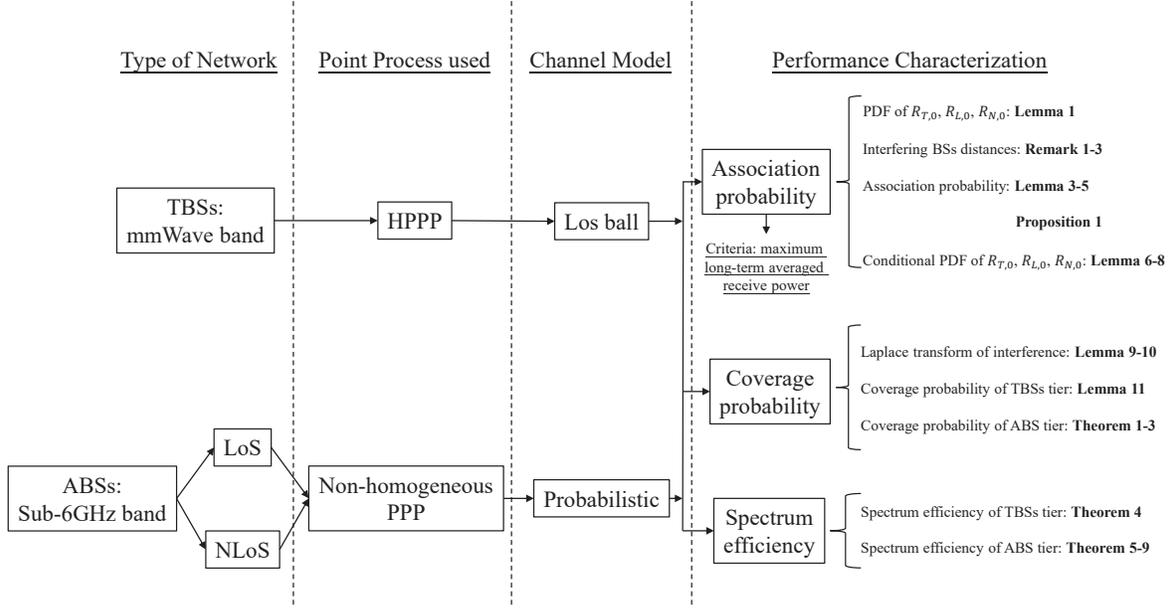}}
		\caption{Flow chart with analytical steps.}\label{flow_chart}
	\end{figure*}

	\section{Network model}
	\subsection{Network Description}
	In this work, we present a two tier downlink UAV-aided hybrid HetNets system shown in Fig. \ref{system_model}. In tier 1, the TBSs  provide wireless connectivity to the ground UEs, where the spatial distribution of TBSs is modeled as an HPPP $\Phi_T$ with density $\lambda_T$.  In tier 2, the ABSs are deployed to enhance the coverage or boost the capacity. We assume that all the ABSs hover at a height $h$ and their horizontal locations form an HPPP $\Phi_A$ with density $\lambda_A$. It is worth mentioning that the analysis in this network model  is also applicable to the ABSs with different altitudes\cite{HetNet_uav1}. Specially, in tier 1 the TBSs are equipped with multiple antennas and  the mmWave band is utilized to provide fast data rate in short-range small cells, while in tier 2 the ABSs adopt sub-6GHz and NOMA techniques in order to improve the coverage and  freedom to serve multiple UEs. All the UEs are assumed to be equipped with two antennas: one to communicate with the TBS on the mmWave channel and the other for the sub-6 GHz connection with the ABS. Without loss of generality, the analysis is conducted based on a  typical UE  positioned at the origin \cite{new_3}.
	All the symbol notations are list in Table \ref{Table of Notations}.
	
	\begin{figure}[!h]
		\centering
		\scalebox{0.3}{\includegraphics{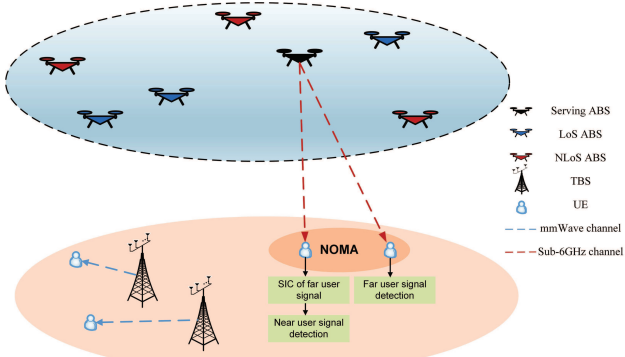}}
		\caption{Illustration of UAV-aided hybrid HetNets.}\label{system_model}
	\end{figure}

	\begin{table}[h!]
		\scriptsize
		\centering
		\caption{Table of Notations}\label{Table of Notations}
		\begin{tabular}{| l | l | }
			\hline
			Notation & Description\\
			\hline
			$\Phi_T$, $\Phi_A$ & HPPP of TBSs, HPPP of ABSs\\
			\hline
			$\Phi_L$, $\Phi_N$ & PPP of LoS ABSs, PPP of NLoS ABSs\\
			\hline
			$a$, $b$& S-curve parameters\\
			\hline
			$P_L(x)$, $P_N(x)$ & \makecell[l]{Probability of LoS/NLoS links under  the horizontal\\ distance $x$} \\
			\hline
			$R_B$, $h$, $R_f$ & mmWave LoS radius, ABS height, fixed UE distance  \\
			\hline
			$R_{N,0}$, $R_{L,0}$, $R_{T,0}$ & \makecell[l]{Minimum distance between the typical UE and the \\NLoS ABSs, LoS ABSs and TBSs}\\
			\hline
			$\alpha_N$, $\alpha_L$, $\alpha_T$  &\makecell[l]{Path loss exponent between the typical UE and the \\NLoS ABSs, LoS ABSs and TBSs}\\
			\hline
			$C_N$, $C_L$, $C_T$ &  \makecell[l]{Additive loss  exponent  between the typical UE and the\\ NLoS ABSs, LoS ABSs and TBSs} \\
			\hline
			$m_N$, $m_L$, $m_T$ & \makecell[l]{Nakagami-$m$ fading parameters  between the typical UE\\ and the NLoS ABSs, LoS ABSs and TBSs}   \\
			\hline
			$a_m$, $a_n$ & NOMA coefficients\\		
			\hline
			$P_T$, $P_A$, $\sigma_T^2$,$\sigma_A^2$  & TBS/ABS transmit power and noise power\\
			\hline
			{  $\beta$}  & {Imperfect SIC coefficient}\\
			\hline
		\end{tabular}
	\end{table}
	
	\subsection{Channel Characteristics and Directional Beamforming in mmWave tier}

	In the first tier of the networks, due to the deployment of mmWave,  the transmitted signals suffer from attenuation due to the obstacles, and the blockage effect can not be  neglected. Here we adopt a tractable equivalent LoS ball model \cite{blockage1} to characterize the blockage effect, which enables fast numerical computation and simplifies the analysis. We define a LoS radius $R_B$, which represents the maximum distance between a UE and its potential mmWave TBS, and the LoS probability is one within $R_B$ and zero outside this radius. It has  been shown that the LoS ball model can fit the real environment properly and provide enough analytical accuracy compared with other blockage models \cite{blockage1}. Regarding the  NLoS paths, it has been pointed out in \cite{blockage3} that the impact of  NLoS signals and NLoS interference can be ignored in mmWave networks. Hence, we will focus on the analysis where the typical UE is associated with a LoS TBS. Then the process  for the case the TBSs  located inside the LoS Ball $\mathcal{B}(0,R_B)$ can be expressed as $\Phi_T\cap\mathcal{B}(0,R_B)$. As a result, the path loss in the TBS tier can be expressed as
	\begin{equation}
		L_T(r)=\mathbf{1}(R_B-r)C_Tr^{-\alpha_T},
	\end{equation}
	where $r$ is the communication distance, $\mathbf{1}(x)$ is  the unit step function, $C_T$ and $\alpha_T$ are the additive loss and path loss exponents, respectively. We also characterize the small scale fading  with the  Nakagami-$m$ fading, where the channel gain follows the Gamma distribution with parameter  $\Gamma(m_T,\frac{1}{m_T})$.

	In this work,  multiple antennas are equipped at the TBSs to accomplish the directional beamforming, and we adopt a 3D sectorized model \cite{directional_beamforming1,directional_beamforming2}. The directivity gain is given by $G(\theta_a,\theta_d)$, where $\theta_a$ is the antenna  3-dB beamwidth for the azimuth orientation in the horizontal direction and $\theta_d$ is the antenna  3-dB beamwidth for the  elevation angles in the TBS, with main-lobe gain $G_M$ and side-lobe gain $G_m$. By adjusting the antenna direction toward the corresponding UE, the UE  benefits from the high main-lobe  gain $G_M$. Moreover, since we evaluate the average directivity gain in our systems, the effect of misalignment
	is ignored in the rest of this paper. Considering the interfering transmission, the beamforming gain and its association probability can be expressed as follows
	\begin{equation}
		G = \left\{ \begin{array}{ll}
			G_M, & p_M=\frac{\theta_a}{2\pi}\cdot\frac{\theta_d}{\pi}\\
			G_m, & p_m=1-\frac{\theta_a}{2\pi}\cdot\frac{\theta_d}{\pi},
		\end{array} \right.\label{gain_probability}
	\end{equation}
	where we assume that the angles deviating from the boresight direction in the azimuth plane $\psi$  is uniformly distributed in the range of $[-\pi,\pi]$ and the deviated  angles  in the depression plane $\phi$ is uniformly distributed in the range of $[-\frac{\pi}{2},\frac{\pi}{2}]$ for all interfering transmissions.
	
	
	\subsection{Channel Characteristics in the NOMA Enabled tier}
	In the second tier of the networks, the  channel between the ABS and the UE is highly affected by the density and altitude of obstacles in the environment, then the A2G channels contain both LoS and NLoS links.  In this paper, we adopt a measurement based on the probabilistic model for LoS/NLoS propagation \cite{ABS_channel}, which is suitable for  sub-6GHz  scenarios.  The probability expressions of  LoS/NLoS links are defined as  $P_L(x)$ and $P_N(x)$, respectively. The LoS link is shown as
	\begin{equation}
		P_L(x)=\frac{1}{1+a\exp(-b(\frac{180}{\pi})\tan^{-1}(\frac{h}{x}-a))},
	\end{equation}
	where $a$ and $b$ are  referred to as the S-curve parameters which are related to the transmission environment, and $x$ denotes the horizontal distance between the ABS and the  UE. Note that a higher ABS altitude results in a higher LoS probability  due to fewer obstacles. Accordingly, the NLoS link probability is given by  $P_N(x)=1-P_L(x).$
	
	Since each link between the ABS and the UE is either in a LoS or NLoS condition with probability $P_L(x)$ and $P_N(x)$, the set of ABSs can be divided into two independent non-homogeneous PPPs $\Phi_L$ and $\Phi_N$ with $\Phi_A=\Phi_L\cup\Phi_N$, which denote the LoS ABSs and the NLoS ABSs, respectively. The corresponding densities of  $\Phi_L$ and $\Phi_N$ with respect to the  horizontal distance $x$ from the typical UE are given by $\lambda_L(x)=2\pi\lambda_A P_L(x)$ and $\lambda_N(x)=2\pi\lambda_A P_N(x)$, respectively. In the following, we use $x_{L,i}$ and $x_{N,i}$ to denote the horizontal locations of the LoS and NLoS ABSs, respectively.
	
	We  consider different channel parameters for LoS/NLoS links in the A2G channels. The additive loss  and path loss exponents of the LoS link in the A2G channels are denoted by $C_L$ and $\alpha_L$, and accordingly we introduce $C_N$ and $\alpha_N$ in the NLoS state.    Therefore, the  channel gains are denoted as  $H_L$ and $H_N$, which   follow the Gamma distribution with parameter $\Gamma(m_L,\frac{1}{m_L})$ and  $\Gamma(m_N,\frac{1}{m_N})$, respectively.

	\subsection{UE Association}
	In this UAV-aided HetNets, a UE is allowed to access the mmWave tier or the NOMA enabled tier in order to provide the best coverage. The flexible UE association is based on the maximum long-term averaged received power at the UE of each tier, which is commonly adopted in other works \cite{HetNet_ref,HetNet_uav6,HetNet_mmWave1}. Intuitively, the typical UE will choose to connect to the BS which has the minimum distance to the UE for both of  the tiers.
	
	\subsubsection{Average Received Power in mmWave  Tier}
	Denoting $R_{T,0}$ as the minimum horizontal distance between the typical UE and the TBSs. We assume that the transmit power of all the TBSs is $P_T$. Thus,  the average received power at the UE connected to the TBS  is given by
	\begin{equation}
		\begin{split}
			P_{r,T}&=G_M P_TC_TR_{T,0}^{-\alpha_T}\mathbf{1}(R_B-R_{T,0})\\
			&\triangleq \eta_T R_{T,0}^{-\alpha_T}\mathbf{1}(R_B-R_{T,0}),
		\end{split}
	\end{equation}
	where  $G_M$ is the directional beamforming gain.
	
	\subsubsection{Average Received Power in the NOMA Enabled Tier }
	
	In the NOMA enabled tier, we adopt UE pairing to implement NOMA  in order to reduce the complexity\cite{UE_pairing}. Compared with the UE association in the OMA scheme, NOMA allocates different power levels to multiple UEs by exploiting power sparsity. The locations of the UEs are also not pre-determined due to the random spatial topology of the stochastic model. As such, we always assume that a near UE is chosen as the typical one first no matter it lies in a LoS/NLoS state\cite{NOMA_reference}, and we denote $R_{L,0}$ and $R_{N,0}$ as the minimum distance between the typical UE and the  LoS/NLoS ABSs. Then  the average received power at the UE connected to the LoS/NLoS ABS can be expressed as
	\begin{equation}
		P_{r,L}=a_nP_AC_LR_{L,0}^{-\alpha_L}\triangleq\eta_L R_{L,0}^{-\alpha_L}\label{power_L},
	\end{equation}
	and
	\begin{equation}
		P_{r,N}=a_nP_AC_NR_{N,0}^{-\alpha_N}\triangleq\eta_N R_{N,0}^{-\alpha_N}\label{power_N},
	\end{equation}
	respectively, where  $P_A$ denotes the transmit power of all the ABSs, and $a_n$ denotes the power allocation factor for the near UE.
	
	
	\subsection{SINR Analysis}
	Due to the fact that the TBS tier and ABS tier  utilize distinctive carrier frequencies,  the signals in these two tiers do not affect each other.
	
	\subsubsection{mmWave Tier Transmission}
	The  SINR at the typical UE when it is connected to a TBS at a distance $R_{T,0}$  can be expressed as
	\begin{equation}
		\gamma_{T}=\frac{G_M P_TC_TR_{T,0}^{-\alpha_T}H_{T,0}}{I_T+\sigma^2_T},
	\end{equation}
	where $R_{T,0}\leq R_B$, and $I_T=\sum_{x_{T,i}\in\Phi_T\cap\mathcal{B}(0,R_B)\backslash x_{T,0}} G_i P_TC_TR_{T,i}^{-\alpha_T}H_{T,i} $ is the interference from the TBS tier. $H_{T,0}$ is the channel gain between the typical UE and the serving TBS, $H_{T,i}$ and $R_{T,i}$ refer to the channel gain and the distance between typical UE and the TBS $i$ (except for the serving BS), respectively. The value and probability of $G_i$ can be obtained through (\ref{gain_probability}). $\sigma^2_T$ denotes the noise power. Both $H_{T,0}$ and $H_{T,i}$ follow the distribution of  $\Gamma(m_T,\frac{1}{m_T})$.
	
	\subsubsection{NOMA Enabled Tier Transmission}
	In the ABS tier, without loss of generality, we consider that each ABS is associated with one UE in the previous round of the UE association process. For simplicity, we follow the assumption in \cite{NOMA_reference} where the distance between the associated UEs and the connected ABSs are the same,  which is an arbitrary value and  denoted as $R_f\geq h$. Since the path loss is more  dominant and stable compared with the small scale fading, we apply the SIC operation at the near UE side. However, it is not pre-determined that the typical UE is the near UE or the far UE, we have the following near UE case and far UE case. We first assume that the typical UE is in a LoS state.
	\paragraph{Near UE in a LoS state case}
	We define ``near UE" as the scenario where the typical UE has smaller distance to the ABS than $R_f$. Then when the typical UE is the near UE in a LoS state, i.e., $R_{L,0}\leq R_f$, the typical UE will first decode the information of the fixed UE to the same LoS ABS with the following SINR
	\begin{equation}
		\gamma_{t\rightarrow f,near}^L=\frac{a_mP_AC_LR_{L,0}^{-\alpha_L}H_{L,0}}{a_nP_AC_LR_{L,0}^{-\alpha_L}H_{L,0}+I_{L}+I_{N}+  \sigma^2_A},\label{t_f_near_LoS}
	\end{equation}
	where $a_m$ denotes the power allocation factor for the far UE which satisfies $a_m>a_n$ and $a_m+a_n=1$. $I_{L}=\sum_{x_{L,i}\in\Phi_L\backslash x_{L,0}} P_AC_L R_{L,i}^{-\alpha_L}H_{L,i}$ denotes the interference from the LoS ABSs and
	$I_{N}=\sum_{x_{N,i}\in\Phi_N} P_AC_NR_{N,i}^{-\alpha_N}H_{N,i}$ denotes the interference from the NLoS ABSs. $H_{L,0}$ is the channel gain between the typical UE and the associated LoS ABS. $H_{L,i}$ and $R_{L,i}$ refer to the channel gain and the  distance between the typical UE and LoS ABS $i$ (except for the serving ABS), respectively. $H_{N,i}$ and $R_{N,i}$ refer to the channel gain and the distance between the typical UE and NLoS ABS $i$, respectively.  $H_{L,0}$ and $H_{L,i}$ follow the distribution of  $\Gamma(m_L,\frac{1}{m_L})$.  $H_{N,i}$ follows the distribution of  $\Gamma(m_N,\frac{1}{m_N})$. $\sigma^2_A$ denotes the noise power of the ABS tier.
	
	If the information of the fixed UE can be decoded successfully, the typical UE will decode its own message with ipSIC coefficient, and the SINR can be expressed as
	\begin{equation}
		\gamma_{t,near}^L=\frac{a_nP_AC_LR_{L,0}^{-\alpha_L}H_{L,0}}{I_{L}+I_{N}+\sigma^2_A+\beta a_mP_AC_LR_{L,0}^{-\alpha_L}H_{L,0}}.\label{t_near_LoS}
	\end{equation}
	where $\beta$ denotes the ipSIC coefficient, which satisfies $0<\beta<1$. When $\beta=0$, the perfect SIC occurs.
	
	For the fixed UE (far UE) served by the same ABS, the signal can be decoded by treating the message of the typical UE as interference, then the SINR for the fixed UE can be expressed as
	\begin{equation}
		\gamma_{ f,near}^L=\frac{a_mP_AC_L R_f^{-\alpha_L} H_{L,f}}{a_nP_AC_L R_f^{-\alpha_L} H_{L,f}+I_{L}+I_{N}+\sigma^2_A},\label{f_near_LoS}
	\end{equation}
	where $H_{L,f}$  refers to the channel gain  between the fixed UE and the serving LoS ABS.
	
	\paragraph{Far UE in a LoS state case}
	On the other hand, when the typical UE has a larger distance to the serving LoS ABS than the fixed UE, i.e., $R_{L,0}> R_f$, the fixed UE will first decode the information of the typical UE with the following SINR
	\begin{equation}
		\gamma_{f\rightarrow t,far}^L=\frac{a_mP_AC_LR_{f}^{-\alpha_L}H_{L,f}}{a_nP_AC_LR_{f}^{-\alpha_L}H_{L,f}+I_{L}+I_{N}+ \sigma^2_A}.
	\end{equation}
	
	Once the information of the typical UE can be decoded successfully, and by applying the SIC technique, the SINR to decode its own message at the fixed UE is given by
	\begin{equation}
		\gamma_{f,far}^L=\frac{a_nP_AC_LR_{f}^{-\alpha_L}H_{L,f}}{I_{L}+I_{N}+\sigma^2_A+\beta a_mP_AC_LR_{f}^{-\alpha_L}H_{L,f}}.
	\end{equation}
	
	For the typical UE  (far UE) that connects to the same LoS ABS, the SINR can be expressed as
	\begin{equation}
		\gamma_{ t,far}^L=\frac{a_mP_AC_L R_{L,0}^{-\alpha_L} H_{L,0}}{a_nP_AC_L R_{L,0}^{-\alpha_L} H_{L,0}+I_{L}+I_{N}+\sigma^2_A}.
	\end{equation}
	
	\paragraph{NLoS state case}
	When the typical UE is associated to a NLoS ABS, the typical UE can be a near UE or a far UE as well.  The SINR analysis in NLoS state is similar to that of the LoS state case,  and we omit the details here.

	\section{Relevant distance distributions and Association Analysis}
	In this section,  we focus on providing the distribution of the distances between the typical UE and the serving TBS, NLoS ABS and LoS ABS, respectively, in the HetNets system.  Furthermore, we derive the expressions for the association probabilities. At last, the distance distributions are characterized given that the typical UE is associated with a TBS, a NLoS ABS or a LoS ABS.
	\subsection{Distance Distributions of the nearest BSs}
	\begin{lemma}\label{lemma_1}
		The probability density function (PDF) of $R_{L,0}$, $R_{N,0}$ and the conditional PDF of $R_{T,0}$ are given by
		\begin{small}
			\begin{align}
				f_{R_{L,0}}(r)=2\pi\lambda_A rP_L(\sqrt{r^2-h^2})&\exp\left(-2\pi\lambda_A\int_{0}^{\sqrt{r^2-h^2}} xP_L(x)dx\right),\nonumber\\&r\geq h
			\end{align}
		\end{small}
		\vspace{-2em}
		\begin{small}
			\begin{align}
				f_{R_{N,0}}(r)=2\pi\lambda_A rP_N(\sqrt{r^2-h^2})&\exp\left(-2\pi\lambda_A\int_{0}^{\sqrt{r^2-h^2}}xP_N(x)dx\right),\nonumber\\ &r\geq h
			\end{align}
		\end{small}
		\vspace{-1.5em}
		\begin{equation}
			f_{R_{T,0}}(r)= \left\{ \begin{array}{ll}
				2\pi\lambda_T r\exp\left(-\pi\lambda_Tr^2\right), & r\leq R_B\\
				0. & r>R_B
			\end{array} \right.
		\end{equation}
	\end{lemma}
	
	\begin{proof}
		Using a similar method to Lemma 1 of \cite{HetNet_ref} and \cite{new_1}, the above expressions can be obtained.
	\end{proof}

	\subsection{Distance of the nearest interfering BSs }
	Due to the deployment of mmWave and the special channel characteristics of A2G channels, the distance of the nearest interfering BSs is not easy to observe. The following remarks show clear insights on the locations of the nearest interfering BSs, which will be useful to derive the main results of this paper.
	
	\subsubsection{The typical UE is associated with a NLoS ABS}
	When the typical UE is associated with a NLoS ABS, we have the following lemma and assumption, which help us to derive the  minimum distance of the interfering LoS ABS.
	
	\begin{lemma}
		The probability that the typical UE has at least one TBS in $\mathcal{B}(0,R_B)$ can be calculated by $Q_T\triangleq F_{R_{T,0}}(R_B)=1-\exp(-\pi\lambda_TR_B^2)$.
	\end{lemma}
	\begin{assumption}
		When there  exists a TBS in $\mathcal{B}(0,R_B)$, the typical UE never associates with a NLoS ABS.
	\end{assumption}
	If there exists a TBS in $\mathcal{B}(0,R_B)$, and  the typical UE is associated with  a NLoS ABS. Then the  average received power at a height $h$ connecting to a NLoS ABS should not be smaller than the average received power at a distance $R_B$ connecting to a TBS, resulting in the condition $h\leq (\frac{\eta_N}{\eta_T})^{\frac{1}{\alpha_N}} R_B^{\frac{\alpha_T}{\alpha_N}}$  satisfied. With a huge path loss and additive loss  in NLoS channels, this condition is not satisfied normally.
	
	\begin{remark}
		Given that the typical UE is associated with a NLoS ABS at a distance $r$, the minimum distance of the interfering LoS ABS  $\tau_{L|N}(r)$ is given by
		\begin{equation}
			\tau_{L|N}(r)=\left(\frac{\eta_L}{\eta_N}\right)^{\frac{1}{\alpha_L}}r^{\frac{\alpha_N}{\alpha_L}}, \quad r\geq h
		\end{equation}
	\end{remark}
	\begin{proof}
		According to (\ref{power_L}) and (\ref{power_N}), the average received power of the UE associated with a NLoS ABS at a distance $r$ is given by $\eta_Nr^{-\alpha_N}$. The minimum distance of the interfering LoS ABS  $\tau_{L|N}(r)$ can be obtained by solving $\eta_Nr^{-\alpha_N}=\eta_L\tau_{L|N}^{-\alpha_L}(r)$.
	\end{proof}
	

	\subsubsection{The typical UE is associated with a LoS ABS}
	
	When the typical UE is associated with a LoS ABS at a distance $r$, depending on whether there exists a TBS in $\mathcal{B}(0,R_B)$, we have two cases. We first denote
	$l_{L,h}\triangleq (\frac{\eta_L}{\eta_N})^{\frac{1}{\alpha_L}}h^{\frac{\alpha_N}{\alpha_L}}$, which is the  distance between the typical UE and a LoS ABS when its average received power is the same as that the average received power from a NLoS ABS  at a distance of $h$,  i.e., $\tau_{L|N}(h)$. We then denote $l_{L,T}\triangleq (\frac{\eta_L}{\eta_T})^{\frac{1}{\alpha_L}} R_B^{\frac{\alpha_T}{\alpha_L}}$ as the \textit{maximum} association distance between the typical UE and a LoS ABS when there exists a TBS in $\mathcal{B}(0,R_B)$. As such, if there exists a TBS in  $\mathcal{B}(0,R_B)$, then only the LoS ABSs which lie in the range of $h\leq r\leq  l_{L,T}$ can be associated to the typical UE.
	
	\begin{remark}
		Given that the typical UE is associated with a LoS ABS at a distance $r$, the minimum distances of the interfering TBS $\tau_{T|L}(r)$ and NLoS ABS $\tau_{N|L}(r)$ are given by
		\begin{equation}
			\tau_{T|L}(r)=\left(\frac{\eta_T}{\eta_L}\right)^{\frac{1}{\alpha_T}}r^{\frac{\alpha_L}{\alpha_T}}, \quad h\leq r\leq
			l_{L,T}.
		\end{equation}
		If there exists a TBS in  $\mathcal{B}(0,R_B)$, and $l_{L,T}\geq l_{L,h}$ is satisfied, we have
		\begin{equation}
			\tau_{N|L}(r)= \left\{ \begin{array}{ll}
				h, & h\leq r\leq l_{L,h}\\
				\left(\frac{\eta_N}{\eta_L}\right)^{\frac{1}{\alpha_N}}r^{\frac{\alpha_L}{\alpha_N}}, & l_{L,h}< r\leq l_{L,T}.
			\end{array} \right.
		\end{equation}
		
		Otherwise when $l_{L,h}> l_{L,T}$, we have
		\begin{equation}
			\tau_{N|L}(r)=h, \quad h\leq r\leq l_{L,T}.
		\end{equation}
		
		If there does not exist a TBS in $\mathcal{B}(0,R_B)$,  $\tau_{N|L}(r)$ can be expressed as
		\begin{equation}
			\tau_{N|L}(r)= \left\{ \begin{array}{ll}
				h, & h\leq r<l_{L,h}\\
				\left(\frac{\eta_N}{\eta_L}\right)^{\frac{1}{\alpha_N}}r^{\frac{\alpha_L}{\alpha_N}}. &  r\geq l_{L,h}
			\end{array} \right.
		\end{equation}
	\end{remark}
	\begin{proof}
		The proof is similar to the proof in Remark 1, therefore  it is omitted here.
	\end{proof}

	\subsubsection{The typical UE is associated with a TBS}
	When the typical UE is associated with a TBS, the condition $0\leq r\leq R_B$ should be satisfied.  We first denote $l_{T,L}\triangleq (\frac{\eta_T}{\eta_L})^{\frac{1}{\alpha_T}} h^{\frac{\alpha_L}{\alpha_T}}$ as the distance between a UE and a TBS when its average received power is the same as  the average received power of  a LoS ABS at a distance of $h$, and similarly we denote $l_{T,N}\triangleq (\frac{\eta_T}{\eta_N})^{\frac{1}{\alpha_T}}h^{\frac{\alpha_N}{\alpha_T}}$.
	
	\begin{remark}
		Given that the typical UE is associated with a TBS at a distance $r$, the minimum distances of the interfering LoS ABS $\tau_{L|T}(r)$ and NLoS ABS $\tau_{N|T}(r)$ are shown below
		
		\begin{table}[h!]
			\centering
			\begin{tabular}{| l | l | l |}
				\hline
				Condition & $\tau_{L|T}(r)$ & $\tau_{N|T}(r)$ \\
				\hline
				$R_B\leq l_{T,L}\leq l_{T,N}$ & $h,\quad 0\leq r\leq R_B$ & $h,\quad 0\leq r\leq R_B$ \\
				\hline
				$l_{T,L} \leq R_B\leq l_{T,N}$ & \tabincell{l}{$h,\quad 0\leq r<l_{T,L}$\\ $\left(\frac{\eta_L}{\eta_T}\right)^{\frac{1}{\alpha_L}}r^{\frac{\alpha_T}{\alpha_L}}$,\\ \quad $l_{T,L}\leq r\leq R_B$} &  $h,\quad 0\leq r\leq R_B$ \\
				\hline
				$R_B\geq l_{T,N} \geq l_{T,L}$ & \tabincell{l}{$h,\quad 0\leq r<l_{T,L}$\\ $\left(\frac{\eta_L}{\eta_T}\right)^{\frac{1}{\alpha_L}}r^{\frac{\alpha_T}{\alpha_L}}$,\\ \quad $l_{T,L}\leq r\leq R_B$} & \tabincell{l}{$h,\quad 0\leq r<l_{T,N}$\\ $\left(\frac{\eta_N}{\eta_T}\right)^{\frac{1}{\alpha_N}}r^{\frac{\alpha_T}{\alpha_N}}$,\\ \quad $l_{T,N}\leq r\leq R_B$}\\
				\hline
			\end{tabular}
		\end{table}
		%
		%
		%
		%
		%
	\end{remark}
	
	\begin{proof}
		The proof is similar to the proof in \textbf{Remark 1} and \textbf{Remark 2}, therefore  it is omitted here.
	\end{proof}
	
	\subsection{Association Probability}
	We first study the probability that  typical UE is associated with a NLoS ABS.
	\begin{lemma}\label{association_N}
		The  probability that a typical UE connects to a NLoS ABS can be calculated as
		\begin{equation}
			\footnotesize
			A_{N}=(1-Q_T)\int_{h}^{\infty}  \exp\left(-2\pi\lambda_A\int_{0}^{\sqrt{\tau^2_{L|N}(r)-h^2}}xP_L(x)dx\right) f_{R_{N,0}}(r) dr.
		\end{equation}
	\end{lemma}
	\begin{proof}
		Using a similar method to Lemma 2 of \cite{HetNet_ref}, and considering  there  does not exist a TBS in $\mathcal{B}(0,R_B)$, the above expressions can be  obtained.
	\end{proof}
	
	
	We then study the probability that the typical UE is associated with a LoS ABS.
	
	\begin{lemma}
		The probability that the typical UE has at least one LoS ABS when $h\leq R_{L,0}\leq l_{L,T}$ can be calculated by $Q_L\triangleq F_{R_{L,0}}(l_{L,T})=1-\exp\left(-2\pi\lambda_A\int_{0}^{\sqrt{l_{L,T}^2-h^2}}xP_L(x)dx\right)$.
	\end{lemma}
	
	\begin{lemma}\label{association_L}
		The  probability that a typical UE connects to a LoS ABS can be calculated as
		\begin{equation}
			A_L = A_{L,1}+A_{L,2},
		\end{equation}
		where $ A_{L,1}$ and $A_{L,2}$ are the association probability related to the cases when there does not exist a TBS and there  exists a TBS in $\mathcal{B}(0,R_B)$, respectively. $ A_{L,1}$ and $A_{L,2}$ are given by
		\begin{small}
			\begin{equation}
				\begin{split}
					&A_{L,1}=(1-Q_T)\Bigg[\int_{h}^{l_{L,h}}f_{R_{L,0}}(r)dr+\int_{l_{L,h}}^{\infty}\\
					&\exp\left(-2\pi\lambda_A\int_{0}^{\sqrt{\tau^2_{N|L}(r)-h^2}}xP_N(x)dx\right) f_{R_{L,0}}(r) dr\Bigg],
				\end{split}
			\end{equation}
		\end{small}
		and
		\begin{equation}
			A_{L,2}= Q_TQ_L\mathbb{P}(A_{L,2}^{N})\mathbb{P}(A_{L,2}^{T}),
		\end{equation}
		respectively, where $\mathbb{P}(A_{L,2}^{N})=1$, if $ l_{L,h}> l_{L,T}$, or
		\begin{small}
			\begin{equation}
				\begin{split}
					\mathbb{P}(A_{L,2}^{N})&=\Bigg[\int_{h}^{l_{L,h}}f_{R_{L,0}}(r)dr+\int_{l_{L,h}}^{l_{L,T}}  \exp\Bigg(-2\pi\lambda_A\\
					&\times\int_{0}^{\sqrt{\tau^2_{N|L}(r)-h^2}}xP_N(x)dx\Bigg) f_{R_{L,0}}(r) dr\Bigg]\ Q_L \label{A_L2_N},
				\end{split}
			\end{equation}
		\end{small}
		and
		\begin{equation}
			\mathbb{P}(A_{L,2}^{T})=\frac{\int_{h}^{l_{L,T}}  \left[\exp\left(-\pi\lambda_T \tau^2_{T|L}(r) \right)-{(1-Q_T)}\right]{f_{R_{L,0}}(r)}}{Q_LQ_T}.
		\end{equation}
	\end{lemma}
	
	\begin{proof}
		See Appendix \ref{Appendix_A}.
	\end{proof}
	
	Finally, we  study the probability that the typical UE is associated with a TBS.
	\begin{proposition}
		The probability that a typical UE connects to a TBS can be calculated as
		\begin{equation}
			A_T = 1-A_N-A_L.
		\end{equation}
	\end{proposition}
	
	\subsection{Conditional distance distribution of the serving BS}
	Now, denoting $\hat{R}_{N,0}$, $\hat{R}_{L,0}$ and $\hat{R}_{T,0}$ as the minimum distance between the typical UE and the serving BS given that it is associated with  the NLoS ABS, LoS ABS and TBS, respectively. The following lemmas characterize the distributions of $\hat{R}_{N,0}$, $\hat{R}_{L,0}$ and $\hat{R}_{T,0}$.
	
	\begin{lemma}\label{conditional_N0}
		The probability density function (PDF) of $\hat{R}_{N,0}$ is given by	
		\begin{equation}\label{conditional_N0_expression}
			\begin{split}
				f_{\hat{R}_{N,0}}(r)&=\frac{f_{R_{N,0}}({r})}{A_N}\exp\Bigg(-2\pi\lambda_A\\
				&\times \int_{0}^{\sqrt{\tau^2_{L|N}(r)-h^2}}xP_L(x)dx\Bigg) (1-Q_T).
			\end{split}
		\end{equation}
	\end{lemma}
	
	\begin{proof}
		Denoting $B_N$ as the event that the typical UE is associated with a NLoS ABS. Then we have
		\begin{equation}
			\begin{split}
				f_{\hat{R}_{N,0}}(r)&=\frac{d}{dr}\mathbb{P}(\hat{R}_{N,0}\leq r)=\frac{d}{dr}\mathbb{P}({R}_{N,0}\leq r|B_N)\\
				&=\frac{d}{dr}\frac{\mathbb{P}({R}_{N,0}\leq r\cap B_N)}{A_N},
			\end{split}
		\end{equation}
		where $\mathbb{P}({R}_{N,0}\leq r\cap B_N)$ can be derived as
		\begin{small}
			\begin{equation}
				\begin{split}
					&\quad \ \mathbb{P}({R}_{N,0}\leq r\cap B_N)\\
					&=\mathbb{P}\left(X_{L,0}>\sqrt{\tau^2_{L|N}(R_{N,0})-h^2}\cap {R}_{N,0}\leq r \right)(1-Q_T)\\
					&=\int_{h}^{r}\exp\left(-2\pi\lambda_A\int_{0}^{\sqrt{\tau^2_{L|N}(\tilde{r})-h^2}}xP_L(x)dx\right) \\
					&\quad \ \times f_{R_{N,0}}(\tilde{r}) d\tilde{r} (1-Q_T).
				\end{split}
			\end{equation}
		\end{small}
		Then taking the first order derivative, we can obtain the expressions in  (\ref{conditional_N0_expression}).
	\end{proof}
	Given that the typical UE is associated with a LoS ABS, depending on whether  there exists a TBS in $\mathcal{B}(0,R_B)$, and utilizing the results in \textbf{Remark 2}, we have the following lemma.
	\begin{lemma}\label{conditional_L0}
		The PDF of $\hat{R}_{L,0}$ is given by	$f_{\hat{R}_{L,0},1}(r)$ or $f_{\hat{R}_{L,0},2}(r)$, where $f_{\hat{R}_{L,0},1}(r)$  is the PDF of  $\hat{R}_{L,0}$  when there does not exist a TBS in $\mathcal{B}(0,R_B)$, and $f_{\hat{R}_{L,0},2}(r)$ is the PDF of  $\hat{R}_{L,0}$  when there exists a TBS in $\mathcal{B}(0,R_B)$. $f_{\hat{R}_{L,0},1}(r)$ and $f_{\hat{R}_{L,0},2}(r)$ are given by
		\begin{equation}
			f_{\hat{R}_{L,0},1}(r)=\left\{ \begin{array}{ll}
				\frac{f_{R_{L,0}}(r)}{ A_{L}} (1-Q_T),\quad h\leq r\leq l_{L,h}\\
				\frac{f_{R_{L,0}}({r})}{ A_{L}}\exp\Bigg(-2\pi\lambda_A\int_{0}^{\sqrt{\tau^2_{N|L}(r)-h^2}}\\
				\times xP_N(x)dx\Bigg)  (1-Q_T), \quad r>l_{L,h}.
			\end{array} \right.
		\end{equation}

		When the condition $l_{L,h}> l_{L,T}$ is satisfied,
		\begin{equation}
			\begin{split}
				f_{\hat{R}_{L,0},2}(r)&=\frac{f_{R_{L,0}}(r)}{ A_{L}}\Big[\exp\Big(-\pi\lambda_T \tau^2_{T|L}(r) \Big)\\
				&-(1-Q_T)\Big],\quad h\leq r\leq l_{L,T}.
			\end{split}
		\end{equation}
		
		When the condition $l_{L,T}\geq l_{L,h}$ is satisfied,
		\begin{equation}
			\small
			f_{\hat{R}_{L,0},2}(r)=\left\{ \begin{array}{ll}
				\frac{f_{R_{L,0}}(r)}{ A_{L}}\left[\exp\left(-\pi\lambda_T \tau^2_{T|L}(r) \right)-(1-Q_T)\right],\\
				\quad\quad\quad\quad\quad\quad\quad\quad\quad\quad\quad\quad\quad h\leq r\leq l_{L,h}\\
				\frac{f_{R_{L,0}}(r)}{A_{L}}\left[\exp\left(-\pi\lambda_T \tau^2_{T|L}(r) \right)-(1-Q_T)\right]\\
				\times\exp\left(-2\pi\lambda_A\int_{0}^{\sqrt{\tau^2_{N|L}(r)-h^2}}xP_N(x)dx\right),\\
				\quad\quad\quad\quad\quad\quad\quad\quad\quad\quad\quad\quad l_{L,h}<r\leq l_{L,T}.
			\end{array} \right.
		\end{equation}
	\end{lemma}
	
	\begin{proof}
		The proof is similar to that in \textbf{Lemma \ref{conditional_N0}}, therefore it is omitted here.
	\end{proof}

	Given that the typical UE is associated with a TBS, and we consider a typical condition  when $l_{T,L} \leq R_B\leq l_{T,N}$. Utilizing the results in \textbf{Remark 3}, we have the following lemma.
	
	\begin{lemma}
		When the condition  $l_{T,L} \leq R_B\leq l_{T,N}$ is satisfied, the PDF of $\hat{R}_{T,0}$ is given by	
	\end{lemma}
	
	\begin{equation}
		\small
		f_{\hat{R}_{T,0}}(r)=\left\{ \begin{array}{ll}
			\frac{f_{R_{T,0}}(r)}{A_{T}},\quad 0\leq r\leq l_{T,L}\\
			\frac{f_{R_{T,0}}({r})}{A_{T}}\exp\left(-2\pi\lambda_A\int_{0}^{\sqrt{\tau^2_{L|T}(r)-h^2}}xP_L(x)dx\right),\\
			\quad\quad\quad\quad\quad\quad\quad\quad\quad\quad\quad\quad\quad \l_{T,L}<r\leq R_B
		\end{array} \right.
	\end{equation}
	\begin{proof}
		The proof is similar to that in \textbf{Lemma \ref{conditional_N0}} and \textbf{Lemma \ref{conditional_L0}}, therefore  it is omitted here.
	\end{proof}

	\section{Coverage Probability Analysis}
	The coverage probability is defined as that a typical UE can successfully transmit signals with a targeted SINR. To begin with, we derive the Laplace transform of the interference.
	
	\subsection{Laplace Transform of Interference }
	Since the ABSs and the TBSs utilize different frequencies, the inter-tier interference can be avoided. We first consider the case that the  interference signals are from the TBSs.
	
	\begin{lemma}\label{laplace_T}
		The Laplace transform of interference from TBSs to a typical UE is given by
		\begin{equation}
			\begin{split}
				L_{I_T}(s)&=\exp\Bigg[\frac{-2\pi\lambda_T}{\alpha_T}\sum_{j\in\{M,m\}}\sum_{i=1}^{m_T}p_j\binom{m_T}{i}\Bigg(\frac{sP_TC_TG_j}{m_T}\Bigg)^{\frac{2}{\alpha_T}}\\
				&\times(-1)^{\frac{2}{\alpha_T}-i+1}
				\times\Bigg(B\Bigg(t_{j,u};i-\frac{2}{\alpha_T},1-m_T\Bigg)\\
				&-B\Bigg(t_{j,l};i-\frac{2}{\alpha_T},1-m_T\Bigg)\Bigg)\Bigg],\label{laplace_TBS}
			\end{split}
		\end{equation}
		where $t_{j,l}=-\frac{sP_TC_TG_j}{m_TR_{T,0}^{\alpha_T}}$,  $t_{j,u}=-\frac{sP_TC_TG_j}{m_TR_{B}^{\alpha_T}}$, and $B(\cdot;\cdot,\cdot)$ is the incomplete Beta function.
	\end{lemma}
	
	\begin{proof}
		See Appendix \ref{Appendix_C}.
	\end{proof}
	
	Then we  consider the scenario where the interference signals are from ABSs.
	
	\begin{lemma}
		The Laplace transform of interference from ABSs to a typical UE is given in (\ref{laplace_IA}) at the top of next page,
		\begin{figure*}[!t]
			\begin{equation}\label{laplace_IA}
				\begin{split}
					L_{I_A}(s)&=\exp\left(-2\pi\lambda_A\int_{w_{L,l}(r)}^{w_{L,u}(r)}\left(1-
					\left(1+\frac{sP_AC_L}{m_L(x^2+h^2)^{\frac{\alpha_L}{2}}}\right)^{-m_L}\right)P_L(x)xdx\right)\\
					&\times \exp\left(-2\pi\lambda_A\int_{w_{N,l}(r)}^{w_{N,u}(r)}\left(1-
					\left(1+\frac{sP_AC_N}{m_N(x^2+h^2)^{\frac{\alpha_N}{2}}}\right)^{-m_N}\right)P_N(x)xdx\right),
				\end{split}
			\end{equation}
			\hrulefill
		\end{figure*}
		where $w_{L,l}(r)$, $w_{L,u}(r)$, $w_{N,l}(r)$, and $w_{N,u}(r)$ are given as follows.
		
		
		\begin{center}
			\small
			\renewcommand\tabcolsep{2.0pt}
			\begin{tabular}{ccccc}
				\hline
				Condition&$w_{L,l}(r)$&$w_{L,u}(r)$&$w_{N,l}(r)$&$w_{N,u}(r)$\\
				\hline
				$B_N$&$\sqrt{\tau_{L|N}^2(r)-h^2}$& $\infty$&$\sqrt{r^2-h^2}$&$\infty$\\
				$B_L$&$\sqrt{r^2-h^2}$&$\infty$&$\sqrt{\tau_{N|L}^2(r)-h^2}$&$\infty$\\
				\hline
			\end{tabular}
		\end{center}
		
	\end{lemma}
	
	\begin{proof}
		The interference from ABSs $I_A$ contains both the interference from LoS ABSs $I_L$ and interference from NLoS ABSs $I_N$, then the Laplace transform of $I_A$ equals to $L_{I_A}(s)=L_{I_L}(s)L_{I_N}(s)$. Following the same steps in \textbf{Lemma \ref{laplace_T}}, we can obtain the above expressions.
	\end{proof}
	
	\subsection{Coverage Probability in mmWave Tier}
	Given that the typical UE is associated with a TBS, the coverage probability is defined as
	\begin{equation}
		P_T^C=\mathbb{P}\left(\frac{G_M P_TC_TR_{T,0}^{-\alpha_T}H_{T,0}}{I_T+\sigma^2}\geq\nu_T\right),
	\end{equation}
	where $\nu_T$ is the targeted SINR of the typical UE.
	\begin{lemma}\label{coverage_T}
		The exact and approximated expression of  $P_T^C$ can be expressed as
		\begin{align}
			&P_T^C=\int_{0}^{R_B}\Bigg(\sum_{k=0}^{m_T-1}\sum_{p=0}^{k}\binom{k}{p}\frac{r^{\alpha_Tk}(m_T\varepsilon_T)^k}{k!}(\sigma^2_T)^p\times\nonumber\\
			&\exp(-m_Tr^{\alpha_T}\varepsilon_T\sigma^2_T)\Bigg[(-1)^{k-p}\frac{\partial^{k-p}}{\partial s^{k-p}}L_{I_T}(s)\Bigg]\Bigg)f_{\hat{R}_{T,0}}(r)dr\nonumber\\
			&\quad\quad\approx\int_{0}^{R_B}\Bigg(\sum_{k=1}^{m_T}(-1)^{k+1}\binom{m_T}{k}e^{-kb_Tr^{\alpha_T}\varepsilon_T\sigma^2_T}\nonumber\\
			&\quad\quad\quad\times L_{I_T}(kb_Tr^{\alpha_T}\varepsilon_T)\Bigg)f_{\hat{R}_{T,0}}(r)dr.\label{T_coverage}
		\end{align}
		where $\varepsilon_T\triangleq\frac{\nu_T}{G_MP_TC_T}$, and ${s=m_T\varepsilon_Tr^{\alpha_T}}$.
	\end{lemma}
	\begin{proof}
		See Appendix \ref{Appendix_D}.
	\end{proof}
	
	It is challenging to obtain exact closed-form expressions for $ P_T^C$. Accordingly, we adopt the Gaussian-Chebyshev quadrature\cite{G_C} to find an approximation of (\ref{T_coverage}) shown in (\ref{cov_T_app}) at the top of next page.
	
	\begin{corollary}
		When the condition  $l_{T,L} \leq R_B\leq l_{T,N}$ is satisfied, (\ref{T_coverage})  can be approximated by
		\begin{figure*}[!t]
			\begin{equation}\label{cov_T_app}
				\begin{split}
					&P_T^C\approx\frac{1}{2A_{T}}\sum_{i=1}^{N}\sum_{k=1}^{m_T}\omega_N(-1)^{k+1}\binom{m_T}{k}\sqrt{1-r_i^2}\Bigg\{l_{T,L}e^{-\xi_kc_i^{\alpha_T} \sigma^2_T}L_{I_T}(\xi_kc_i^{\alpha_T})f_{R_{T,0}}(c_i)+\\
					&(R_B-l_{T,L})e^{-\xi_kd_i^{\alpha_T}  \sigma^2_T}L_{I_T}(\xi_kd_i^{\alpha_T})f_{R_{T,0}}(d_i)\exp\left(-2\pi\lambda_A\int_{0}^{\sqrt{\tau^2_{L|T}(d_i)-h^2}}xP_L(x)dx\right)\Bigg\},
				\end{split}
			\end{equation}
			\hrulefill
		\end{figure*}
		where $\xi_k=kb_T\varepsilon_T$, and $N$ is a parameter to ensure a complexity-accuracy tradeoff, $\omega_N=\frac{\pi}{N}$, $r_i=\cos\left(\frac{2i-1}{2n}\pi\right)$, $c_i=\frac{l_{T,L}}{2}r_i+\frac{l_{T,L}}{2}$, and $d_i=\frac{R_B-l_{T,L}}{2}r_i+\frac{R_B-l_{T,L}}{2}$.
	\end{corollary}
	
	\subsection{Coverage Probability in the NOMA Enabled Tier}
	According to the NOMA decoding strategy, two cases will be considered, i.e. the near UE case and the far UE case. We first consider that the typical UE is in the LoS state.
	\subsubsection{LoS Near UE Case}
	In this case, i.e,  $R_{L,0}\leq R_f$, when the following conditions hold, successful decoding will occur:
	
	$\bullet$ The typical UE can decode the information of the fixed UE served by the same BS.
	
	$\bullet$ After the SIC process, the typical UE can decode its own information.
	
	Thus, the coverage probability in this case can be expressed as
	\begin{equation}
		\small
		P_{cov,near|LoS}(R_{L,0})=\mathbb{P}(\gamma_{t\rightarrow f,near}^L>\epsilon_f, \gamma_{t,near}^L>\epsilon_t),\label{cov_near_LoS}
	\end{equation}
	where $\epsilon_f$ and $\epsilon_t$ are the targeted SINR of the fixed UE and the typical UE, respectively.
	
	Based on (\ref{cov_near_LoS}), the  coverage probability of a typical UE for the near UE case in a LoS state is given in the following Proposition.
	\begin{proposition}\label{coverage_near_L}
		If $a_m-a_n\epsilon_f\geq 0$ and $a_n-\beta a_m\epsilon_t\geq 0$ hold, the exact and approximated  coverage probability of a typical UE for the LoS UE case can be expressed as
		{\small
			\begin{align}\label{cov_near_LoS_expression}
				&P_{cov,near|LoS}(R_{L,0})\nonumber\\
				&=\sum_{k=0}^{m_L-1}\sum_{p=0}^k\binom{k}{p}\frac{R_{L,0}^{\alpha_Lk}(m_L\epsilon^*)^k}{k!(P_AC_L)^k}\exp\left(-m_L\frac{\epsilon^*R_{L,0}^{\alpha_L}}{P_AC_L}\sigma^2_A\right)\nonumber\\
				&\quad\quad\times(  \sigma^2_A)^p\mathbb{E}_{I_A}\left[\exp\left(-m_L\frac{\epsilon^*R_{L,0}^{\alpha_L}}{P_AC_L}I_A\right)(I_A)^{k-p}\right]\nonumber\\
				&\approx\sum_{k=1}^{m_L}(-1)^{k+1}\binom{m_L}{k}e^{-kb_L\frac{\epsilon^*R_{L,0}^{\alpha_L} \sigma^2_A}{P_AC_L}}L_{I_A}\left(kb_L\frac{\epsilon^*R_{L,0}^{\alpha_L}}{P_AC_L}\right),
		\end{align}}
		where  $\epsilon^*=\max\left\{\frac{\epsilon_f}{a_m-a_n\epsilon_f},\frac{\epsilon_t}{a_n-\beta a_m \epsilon_t}\right\}$ and  $b_L=m_L(m_L!)^{-\frac{1}{m_L}}.$
		Otherwise $P_{cov,near|LoS}(R_{L,0})=0$.
	\end{proposition}
	
	\begin{proof}
		Substituting (\ref{t_f_near_LoS}) and (\ref{t_near_LoS}) into (\ref{cov_near_LoS}), and following the same steps in \textbf{Appendix \ref{Appendix_D}}, we can obtain the exact expression of $P_{cov,near|LoS}(R_{L,0})$. Similarly, by adopting the Alzer's Lemma\cite{gamma_bound}, we obtain the results in (\ref{cov_near_LoS_expression}).
	\end{proof}
	
	\subsubsection{LoS Far UE Case}
	For this case, i.e,  $R_{L,0}> R_f$,  successful decoding will occur if the typical UE can decode its own information by treating the fixed UE as noise.  The  coverage probability of a typical UE for the far UE case in a LoS state is given in the following proposition.
	
	\begin{proposition}\label{coverage_far_L}
		If $a_m-a_n\epsilon_t\geq 0$ holds, the coverage probability of a typical UE for the LoS UE case can be expressed as
		\begin{align}\label{cov_far_LoS_expression}
			P_{cov,far|LoS}(R_{L,0})&\approx\sum_{k=1}^{m_L}(-1)^{k+1}\binom{m_L}{k}e^{-kb_L\frac{\epsilon_t^fR_{L,0}^{\alpha_L}\sigma^2_A}{P_AC_L}}\nonumber\\
			&\quad\times L_{I_A}\left(kb_L\frac{\epsilon_t^fR_{L,0}^{\alpha_L}}{P_AC_L}\right),
		\end{align}
		otherwise $P_{cov,far|LoS}(R_{L,0})=0$,
		where  $\epsilon_t^f=\frac{\epsilon_t}{a_m-a_n\epsilon_t}$.
	\end{proposition}
	\begin{proof}
		By following the similar procedure, with interchanging $\epsilon^*$ with $\epsilon_t^f$, we obtain the desired results.
	\end{proof}
	\begin{remark}
		\textbf{Proposition \ref{coverage_near_L}} and \textbf{Proposition \ref{coverage_far_L}} provide  the basic power allocation guidelines for
		the NOMA enabled ABSs networks. The targeted SINR threshold of the typical UE and  the fixed UE both determine the coverage probability of a typical UE when it is associated with an ABS. Furthermore, inappropriate power allocation such as $a_m-a_n\epsilon_f<0$, $a_m-a_n\epsilon_t<0$ and  $a_n-\beta a_m\epsilon_t<0$ cause the coverage probability in  the NOMA enabled tier always being zero.
	\end{remark}
	\begin{theorem}\label{theorem1}
		The coverage probability of a typical UE associated with a LoS ABS is given by
		\begin{align}\label{cov_LoS_final}
			&P_{cov,LoS}(\epsilon_f,\epsilon_t)=\sum_{i=1}^{2}\Bigg[\int_{h}^{R_f}P_{cov,near|LoS}(r) f_{\hat{R}_{L,0},i}(r)dr\nonumber\\
			&\quad\quad\quad\quad+\int_{R_f}^{\infty}P_{cov,far|LoS}(r) f_{\hat{R}_{L,0},i}(r)dr\Bigg].
		\end{align}
	\end{theorem}
	
	\begin{proof}
		Based on (\ref{cov_near_LoS_expression}) and (\ref{cov_far_LoS_expression}), and  by considering the distance distribution of a typical UE associated with a LoS ABS, the result in (\ref{cov_LoS_final}) can be easily obtained.
	\end{proof}
	
	\subsubsection{NLoS Case}
	Then we consider the case that the typical UE is associated with a NLoS ABS. By following  the proof in \textbf{Proposition \ref{coverage_near_L}}, we can obtain
	\begin{align}
		P_{cov,near|NLoS}(R_{N,0})&\approx\sum_{k=1}^{m_N}(-1)^{k+1}\binom{m_N}{k}e^{-kb_N\frac{\epsilon^*R_{N,0}^{\alpha_N}\sigma^2_A}{P_AC_N}}\nonumber\\
		&\quad\times L_{I_A}\left(kb_N\frac{\epsilon^*R_{N,0}^{\alpha_N}}{P_AC_N}\right),
	\end{align}
	where $b_N=m_N(m_N!)^{-\frac{1}{m_N}}$, and
	\begin{align}
		P_{cov,far|NLoS}(R_{N,0})&\approx\sum_{k=1}^{m_N}(-1)^{k+1}\binom{m_N}{k}e^{-kb_N\frac{\epsilon_t^fR_{N,0}^{\alpha_N}\sigma^2_A}{P_AC_N}}\nonumber\\
		&\quad\times L_{I_A}\left(kb_N\frac{\epsilon_t^fR_{N,0}^{\alpha_N}}{P_AC_N}\right).
	\end{align}
	\begin{theorem}
		The  coverage probability of a typical UE associated with a NLoS ABS is given by
		\begin{align}\label{cov_NLoS_final}
			P_{cov,NLoS}(\epsilon_f,\epsilon_t)&=\int_{h}^{R_f}P_{cov,near|NLoS}(r) f_{\hat{R}_{N,0}}(r)dr\nonumber\\
			&+\int_{R_f}^{\infty}P_{cov,far|NLoS}(r) f_{\hat{R}_{N,0}}(r)dr.
		\end{align}
	\end{theorem}
	
	\begin{proof}
		The proof is similar  to the proof in \textbf{Theorem \ref{theorem1}}.
	\end{proof}
	\begin{theorem}
		The coverage probability of a typical UE associated to the NOMA enabled ABS tier is given by
		\begin{equation} \label{cov_ABS_final}
			P_A^C(\epsilon_f,\epsilon_t)=\frac{A_L}{A_A}P_{cov,LoS}(\epsilon_f,\epsilon_t)+\frac{A_N}{A_A}P_{cov,NLoS}(\epsilon_f,\epsilon_t).
		\end{equation}
	\end{theorem}
	\begin{proof}
		Based on (\ref{cov_LoS_final}) and (\ref{cov_NLoS_final}), by considering that the typical UE is associated with an ABS, the result in (\ref{cov_ABS_final}) can be  obtained.
	\end{proof}
	
	\section{Spectrum efficiency}
	In this section, we evaluate the spectrum efficiency of the proposed HetNets.
	\subsection{Ergodic Rate of mmWave Tier}
	\begin{theorem}
		The achievable ergodic rate of the mmWave tier can be expressed as follows
		\begin{equation}
			R_T=\frac{1}{\ln 2}\int_{0}^{\infty}\frac{\bar{F}_{\gamma_T}(z)}{1+z}dz,
		\end{equation}
		where $\bar{F}_{\gamma_T}(z)$ is approximated by
		\begin{align}
			\bar{F}_{\gamma_T}(z)&{\approx}\int_{0}^{R_B}\Bigg(\sum_{k=1}^{m_T}(-1)^{k+1}\binom{m_T}{k}e^{-kb_T\frac{zr^{\alpha_T}\sigma^2_T}{G_MP_TC_T}}\nonumber\\
			&\quad \times L_{I_T}\left(kb_T\frac{zr^{\alpha_T}}{G_MP_TC_T}\right)\Bigg)f_{\hat{R}_{T,0}}(r)dr.
		\end{align}
	\end{theorem}
	\begin{proof}
		For the UE connected to a TBS, the achievable ergodic rate can be expressed as
		\begin{align}
			R_T&=\mathbb{E}\left[\log_2(1+\gamma_{T})\right]\nonumber\\
			&=\int_{0}^{\infty}\mathbb{P}\left(\gamma_{T}>z\right)d\log_2\left(1+z\right)=\frac{1}{\ln 2}\int_{0}^{\infty}\frac{\bar{F}_{\gamma_T}(z)}{1+z}dz.
		\end{align}
		Then, given the typical UE is associated with a TBS, and following the same procedure in \textbf{Appendix \ref{Appendix_D}}, the expression for $\bar{F}_{\gamma_T}(z)$ is given by
		\begin{align}
			\bar{F}_{\gamma_T}(z)&=\mathbb{P}\left(\frac{G_M P_TC_TR_{T,0}^{-\alpha_T}H_{T,0}}{I_T+\sigma_T^2}>z\right)\nonumber\\
			&{\approx}\int_{0}^{R_B}\Bigg(\sum_{k=1}^{m_T}(-1)^{k+1}\binom{m_T}{k}e^{-kb_T\frac{zr^{\alpha_T}{  \sigma^2_T}}{G_MP_TC_T}}\nonumber\\
			&\quad \times L_{I_T}\left(kb_T\frac{zr^{\alpha_T}}{G_MP_TC_T}\right)\Bigg)f_{\hat{R}_{T,0}}(r)dr.
		\end{align}
	\end{proof}
	\subsection{Ergodic Rate of the NOMA Enabled Tier}
	Different from the derivations of ergodic rate in mmWave tier, the achievable ergodic rate for the NOMA enabled tier is determined by the channel conditions of UEs. If the far UE can decode its own message, the near UE  can decode the message of the far UE since it has a better channel condition.  Denote $f_1(z)\triangleq=\frac{z}{a_n-\beta a_m z}$ and $f_2(z)=\frac{z}{a_m-a_nz}$, then the following theorems show the ergodic rate of the NOMA enabled tier.
	
	\begin{theorem}
		When the typical UE is associated with an ABS, and	for the LoS near UE case, the achievable ergodic rate can be expressed as
		\begin{equation}
			\small
			R_{near|LoS}=\frac{1}{\ln 2}\int_{0}^{\frac{a_n}{\beta a_m}}\frac{\bar{F}_{\gamma_{t,near}^L}(z)}{1+z}dz+\frac{1}{\ln 2}\int_{0}^{\frac{a_m}{a_n}}\frac{\bar{F}_{\gamma_{f,near}^L}(z)}{1+z}dz,
		\end{equation}
		where
		\begin{align}
			\bar{F}_{\gamma_{t,near}^L}(z)&\approx\sum_{i=1}^{2}\int_{h}^{R_f}\Bigg[\sum_{k=1}^{m_L}(-1)^{k+1}\binom{m_L}{k}e^{-f_1(z)\frac{kb_Lr^{\alpha_L}\sigma^2_A}{P_AC_L}}\nonumber\\
			&\quad\times L_{I_A}\left( f_1(z)\frac{kb_Lr^{\alpha_L}}{P_AC_L}\right)\Bigg]f_{\hat{R}_{L,0},i}(r)dr,
		\end{align}
		and
		\begin{align}
			\bar{F}_{\gamma_{f,near}^L}(z)&{\approx}\sum_{i=1}^{2}\int_{h}^{R_f}\Bigg[ \sum_{k=1}^{m_L}(-1)^{k+1}\binom{m_L}{k}e^{-f_2(z)\frac{kb_LR_{f}^{\alpha_L} \sigma^2_A}{P_AC_L}}\nonumber\\
			&\quad\times L_{I_A}\left( f_2(z)\frac{kb_LR_{f}^{\alpha_L}}{P_AC_L}\right)\Bigg]f_{\hat{R}_{L,0},i}(r)dr.
		\end{align}
	\end{theorem}

\begin{proof}
	When the typical UE is associated with an ABS, and	for the LoS near UE case, the achievable ergodic rate can be expressed as
	\begin{align}
		&R_{near|LoS}=\mathbb{E}\left[\log_2(1+\gamma_{t,near}^L)+\log_2(1+\gamma_{f,near}^L)\right]\nonumber\\
		&=\frac{1}{\ln 2}\int_{0}^{\infty}\frac{\bar{F}_{\gamma_{t,near}^L}(z)}{1+z}dz+\frac{1}{\ln 2}\int_{0}^{\infty}\frac{\bar{F}_{\gamma_{f,near}^L}(z)}{1+z}dz.
	\end{align}
	Based on (\ref{t_near_LoS}), we can obtain the expressions for $\bar{F}_{\gamma_{t,near}^L}(z)$  as follows
	\begin{align}
		\bar{F}_{\gamma_{t,near}^L}(z)&=\mathbb{P}\left(\frac{P_AC_LR_{L,0}^{-\alpha_L}H_{L,0}}{I_{L}+I_{N}+\sigma^2}>{  f_1(z)}\right)\nonumber\\
		&{\approx}\ \mathbb{E}_{R_{L,0}}\Bigg[\sum_{k=1}^{m_L}(-1)^{k+1}\binom{m_L}{k}e^{{  -f_1(z)}\frac{kb_LR_{L,0}^{\alpha_L}{  \sigma^2_A}}{P_AC_L}}\nonumber\\
		&\quad\times L_{I_A}\left({  f_1(z)}\frac{kb_LR_{L,0}^{\alpha_L}}{P_AC_L}\right)\Bigg].
	\end{align}
	Note that for the case $z\geq\frac{a_n}{\beta a_m}$, it is easy to observe that $\bar{F}_{\gamma_{t,near}^L}=0$.	
	
	Then according to whether there exists a TBS in $\mathcal{B}(0,R_B)$, we can obtain the above expressions.

	Similarly, 	based on (\ref{f_near_LoS}), we can obtain the expressions for $\bar{F}_{\gamma_{f,near}^L}$  as follows
	\begin{align}
		\bar{F}_{\gamma_{f,near}^L}(z)&=\mathbb{P}\left(H_{L,f}>{  f_2(z)}\frac{R_f^{\alpha_L}(I_A+\sigma^2)}{P_AC_L}\right)\nonumber\\
		&{\approx}\ \mathbb{E}_{R_{L,0}}\Bigg[ \sum_{k=1}^{m_L}(-1)^{k+1}\binom{m_L}{k}e^{{  -f_2(z)}\frac{kb_LR_{f}^{\alpha_L}{  \sigma^2_A}}{P_AC_L}}\nonumber\\
		&\quad\times L_{I_A}\left({  f_2(z)}\frac{kb_LR_{f}^{\alpha_L}}{P_AC_L}\right)\Bigg].
	\end{align}
	
	Note that for the case $z\geq\frac{a_m}{a_n}$, it is easy to observe that $\bar{F}_{\gamma_{f,near}^L}=0$.
\end{proof}

By following the same procedures, we can obtain the theorems below.

\begin{theorem}
	When the typical UE is associated with an ABS, and	for the LoS far UE case, the achievable ergodic rate can be expressed as
	\begin{equation}
		R_{far|LoS}=\frac{1}{\ln 2}\int_{0}^{{  \frac{a_n}{\beta a_m}}}\frac{\bar{F}_{\gamma_{f,far}^L}(z)}{1+z}dz+\frac{1}{\ln 2}\int_{0}^{\frac{a_m}{a_n}}\frac{\bar{F}_{\gamma_{t,far}^L}(z)}{1+z}dz,
	\end{equation}
	where 
	\begin{align}
		\bar{F}_{\gamma_{t,far}^L}(z)&\approx\sum_{i=1}^{2}\int_{R_f}^{\infty}\Bigg[\sum_{k=1}^{m_L}(-1)^{k+1}\binom{m_L}{k}e^{{  -f_2(z)}\frac{kb_Lr^{\alpha_L}{  \sigma^2_A}}{P_AC_L}}\nonumber\\
		&\quad\times 
		L_{I_A}\left({  f_2(z)}\frac{kb_Lr^{\alpha_L}}{P_AC_L}\right)\Bigg]f_{\hat{R}_{L,0},i}(r)dr,
	\end{align}
	and
	\begin{align}
		\bar{F}_{\gamma_{f,far}^L}(z)&\approx\sum_{i=1}^{2}\int_{R_f}^{\infty}\Bigg[\sum_{k=1}^{m_L}(-1)^{k+1}\binom{m_L}{k}e^{{  -f_1(z)}\frac{kb_LR_f^{\alpha_L}{  \sigma^2_A}}{P_AC_L}}\nonumber\\
		&\quad\times L_{I_A}\left({  f_1(z)}\frac{kb_LR_f^{\alpha_L}}{P_AC_L}\right)\Bigg]f_{\hat{R}_{L,0},i}(r)dr.
	\end{align}
\end{theorem}

\begin{theorem}
	When the typical UE is associated with an ABS, and	for the NLoS near UE case, the achievable ergodic rate can be expressed as
	\begin{equation}
		\small
		R_{near|NLoS}=\frac{1}{\ln 2}\int_{0}^{{  \frac{a_n}{\beta a_m}}}\frac{\bar{F}_{\gamma_{t,near}^N}(z)}{1+z}dz+\frac{1}{\ln 2}\int_{0}^{\frac{a_m}{a_n}}\frac{\bar{F}_{\gamma_{f,near}^N}(z)}{1+z}dz,
	\end{equation}
	where 
	\begin{align}
		\bar{F}_{\gamma_{t,near}^N}(z)&\approx \int_{h}^{R_f}\Bigg[\sum_{k=1}^{m_N}(-1)^{k+1}\binom{m_N}{k}e^{{  - f_1(z)}\frac{kb_Nr^{\alpha_N}{  \sigma^2_A}}{P_AC_N}}\nonumber\\
		&\quad\times L_{I_A}\left({  f_1(z)}\frac{kb_Nr^{\alpha_N}}{P_AC_N}\right)\Bigg]f_{\hat{R}_{N,0}}(r)dr,
	\end{align}
	and
	\begin{align}
		\bar{F}_{\gamma_{f,near}^N}(z)&\approx\int_{h}^{R_f}\Bigg[\sum_{k=1}^{m_N}(-1)^{k+1}\binom{m_N}{k}e^{{  -f_2(z)}\frac{kb_NR_{f}^{\alpha_N}{  \sigma^2_A}}{P_AC_N}}\nonumber\\
		&\quad\times L_{I_A}\left({  f_2(z)}\frac{kb_NR_{f}^{\alpha_N}}{P_AC_N}\right)\Bigg]f_{\hat{R}_{N,0}}(r)dr.
	\end{align}
	
\end{theorem}
\begin{theorem}
	When the typical UE is associated with an ABS, and	for the NLoS far UE case, the achievable ergodic rate can be expressed as
	\begin{equation}
		\small
		R_{far|NLoS}=\frac{1}{\ln 2}\int_{0}^{{  \frac{a_n}{\beta a_m}}}\frac{\bar{F}_{\gamma_{f,far}^N}(z)}{1+z}dz+\frac{1}{\ln 2}\int_{0}^{\frac{a_m}{a_n}}\frac{\bar{F}_{\gamma_{t,far}^N}(z)}{1+z}dz,
	\end{equation}
	where	
	\begin{align}
		\bar{F}_{\gamma_{t,far}^N}(z)&\approx \int_{R_f}^{\infty}\Bigg[\sum_{k=1}^{m_N}(-1)^{k+1}\binom{m_N}{k}e^{{  -f_2(z)}\frac{kb_Nr^{\alpha_N}{  \sigma^2_A}}{P_AC_N}}\nonumber\\
		&\quad\times L_{I_A}\left({  f_2(z)}\frac{kb_Nr^{\alpha_N}}{P_AC_N}\right)\Bigg]f_{\hat{R}_{N,0}}(r)dr,
	\end{align}
	and
	\begin{align}
		\bar{F}_{\gamma_{f,far}^N}(z)&\approx\int_{R_f}^{\infty}\Bigg[\sum_{k=1}^{m_N}(-1)^{k+1}\binom{m_N}{k}e^{{  -f_1(z)}\frac{kb_NR_f^{\alpha_N}{  \sigma^2_A}}{P_AC_N}}\nonumber\\
		&\quad\times L_{I_A}\left({  f_1(z)}\frac{kb_NR_f^{\alpha_N}}{P_AC_N}\right)\Bigg]f_{\hat{R}_{N,0}}(r)dr.
	\end{align}	
\end{theorem}

\begin{theorem}
	The achievable ergodic rate of the NOMA enabled tier can be expressed as
	\begin{equation}\label{rate_A}
		\small
		R_A = \frac{A_L}{A_A}\left(R_{near|LoS}+R_{far|LoS}\right)+\frac{A_N}{A_A}\left(R_{near|NLoS}+R_{far|NLoS}\right).
	\end{equation}
\end{theorem}

\begin{proof}
	Based on \textbf{Theorems 5-8}, and  due to the fact that the typical UE is associated with an ABS, the result in (\ref{rate_A}) can be  obtained. 
\end{proof}

	\section{Simulation results}

	In this section, numerical results are provided to facilitate the performance evaluations of  the proposed UAV-aided HetNets.   100,000 times Monte Carlo simulations are  carried out to verify the accuracy of the  analytical expressions. All the horizontal locations of  ABSs and TBSs are distributed in a disc with radius  $5\times10^4$  m. Other  parameters are summarized in Table \ref{Table of Parameters}.
	
	\begin{table}[h!]
		\footnotesize
		\renewcommand\tabcolsep{2.0pt}
		\centering
		\caption{Table of Parameters}\label{Table of Parameters}
		\begin{tabular}{| l | l || l | l | }
			\hline
			\textbf{Parameter}  & \textbf{Value}  & \textbf{Parameter} & \textbf{Value} \\
			\hline
			$\lambda_T$ & $10^{-5}/m^2$ & $\lambda_A$ & $0.2\lambda_T$\\
			\hline
			$(a$, $b)$ & (12.08, 0.11)  & $h$,  $R_f$& 200 m, $1.1h$ \\
			\hline
			$R_B$  & 220 m  &  $\beta$  & 0.1\\
			\hline
			$(\alpha_N,\alpha_L,\alpha_T)$ & (3, 2.5, 2) & ($C_N, C_L, C_T$) & (10, 3, 3) dB\\
			\hline
			$(m_N, m_L, m_T)$ & (1, 2, 2) & ($a_m, a_n$) & (0.8, 0.2) \\
			\hline
			$(P_T, P_A)$ & (20, 59) dBm & ($\sigma^2_T, \sigma^2_A$) & ($-70$, $-104$) dBm\\
			\hline
			$N_T$ & $4$ & ($\theta_a, \theta_d$) & ($\sqrt{3/N_T}, \sqrt{3/N_T}$)\\
			\hline
			$G_M$  & $N_T$ & $G_m$ & $\frac{\sqrt{N_T}-\frac{\sqrt{3}N_T\sin\left(\frac{3\pi}{2\sqrt{N_T}}\right)}{2\pi}}{\sqrt{N_T}-\frac{\sqrt{3}\sin\left(\frac{3\pi}{2\sqrt{N_T}}\right)}{2\pi}}$\\
			\hline
		\end{tabular}
	\end{table}

	First, we evaluate the association performance of the UAV-aided HetNets. In Fig. \ref{fig_1}, for a given set of ABSs/TBSs densities, the solid curves and dashed curves are the association probability for TBSs and ABSs, respectively. It shows the association probability versus the altitude of ABSs, and the simulation results and analytical results match perfectly. We can observe that, at low altitudes, the probability that a typical UE is associated with an ABS is small due to the fact it experiences a large fraction of NLoS A2G links, which results in a large path loss. Then with the increase of the altitude of ABSs, the probability that a typical UE is associated with an ABS begins to increase.  This is because the channel condition between the typical UE and the ABSs becomes better, i.e., the fraction of the LoS A2G links increases. When the altitude of the ABSs further increases, although most of the A2G links experience the LoS condition, the huge distance between the typical UE and the associated ABS will result in a non-negligible path loss. We can also notice that when the ratio of densities of ABSs and TBSs are fixed,  a typical UE will tend to associate with a TBS when  their densities are increased. This is rather intuitive due to the fact that the path loss exponent in mmWave transmissions is small.
	\begin{figure}[!ht]
		\centering
		\scalebox{0.35}{\includegraphics{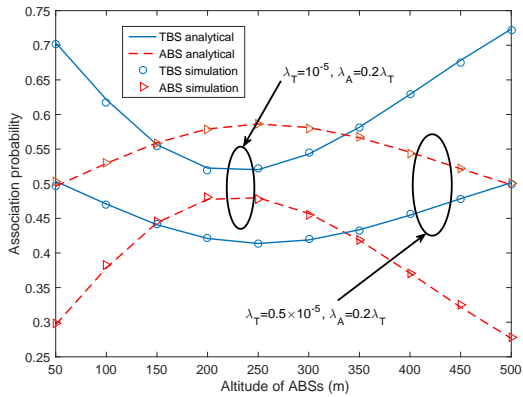}}
		\caption{Probability of association to an ABS/TBS versus altitude of ABSs for different densities }\label{fig_1}
	\end{figure}
	
	We then investigate the coverage performance of the UAV-aided HetNets. Fig. \ref{fig_3} shows the coverage probability achieved by the typical UE  when it is associated with a TBS. From the figure, we can observe that the analytical  results match the simulation results well. Moreover, the analytical results are slightly larger than that of the simulation results. This is because we adopt the Alzer's Lemma to provide an approximation  of a Gamma random variable when $m_T=2$. We also plot the approximated analytical results shown in (\ref{cov_T_app}) in \textbf{Corollary 1}.  The closed-form expression by adopting the Gaussian-Chebyshev quadrature has a relatively exact value compared with the simulation results. From Fig. \ref{fig_3}, we also notice that due to the use of mmWave, the typical UE can achieve a large coverage probability, and with the increase of the number of  antennas, the coverage probability will increase significantly even at a large SINR threshold.
	
	%
	
	\begin{figure}[!ht]
		\centering
		\scalebox{0.35}{\includegraphics{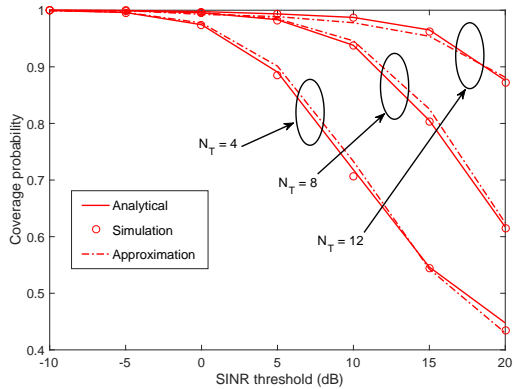}}
		\caption{TBS tier (mmWave tier) coverage probability versus SINR threshold, $h=200$ m}\label{fig_3}
	\end{figure}
	
	\begin{figure}[!ht]
		\centering
		\scalebox{0.35}{\includegraphics{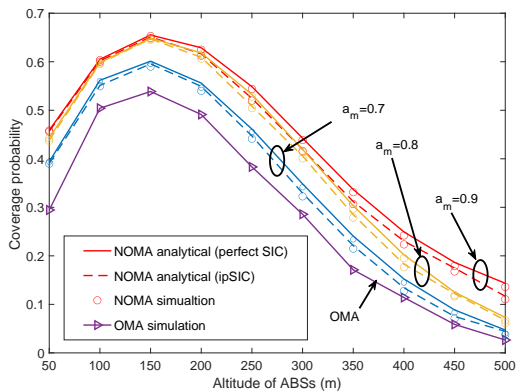}}
		\caption{NOMA enabled ABS tier coverage probability versus altitude of ABSs, $\epsilon_f=\epsilon_t=0$ dB}\label{fig_4}
	\end{figure}
	
	Fig. \ref{fig_4} plots the coverage probability of a typical UE versus the altitude of ABSs when it is associated with an ABS both in the NOMA  and OMA schemes, and the impact of imperfect SIC is studied. Note that the performance of the OMA scheme is only shown through the numerical approach, and it is adopted by dividing these two UEs in equal time/frequency slots.
	Similar to the results in Fig. \ref{fig_1}, the coverage probability will first increase then decrease due to the A2G  channel characteristics. We also demonstrate the superiority of NOMA over OMA  when the power allocation factors are selected appropriately under some SINR threshold. The imperfect SIC will reduce the performance the NOMA assisted network.
	It is worth mentioning that the power allocation between these two NOMA UEs can affect the coverage probability  significantly. However,  the optimization of the power allocation is beyond the scope of this paper.

	\begin{figure}[!ht]
		\centering
		\scalebox{0.35}{\includegraphics{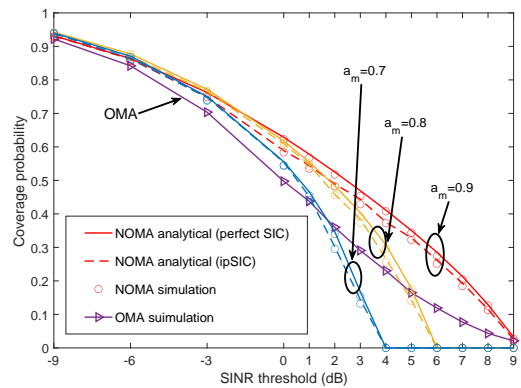}}
		\caption{NOMA enabled ABS tier  coverage probability versus SINR threshold, $h=200$ m}\label{fig_5}
	\end{figure}
	
	\begin{figure}[!ht]
		\centering
		\scalebox{0.38}{\includegraphics{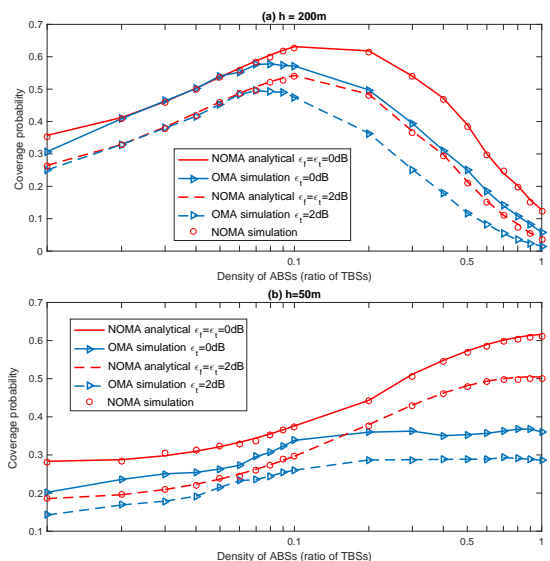}}
		\caption{NOMA enabled ABS tier coverage probability versus density of ABSs}\label{fig_7}
	\end{figure}
	
	\begin{figure}[!ht]
		\centering
		\scalebox{0.36}{\includegraphics{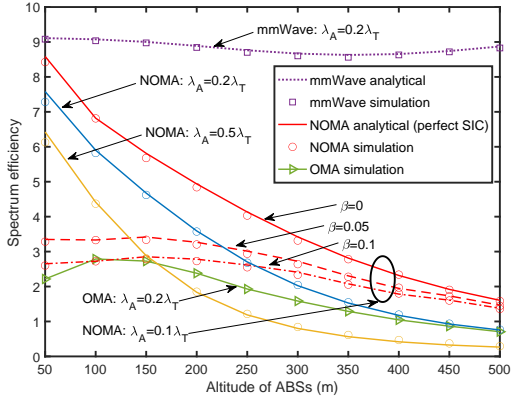}}
		\caption{Spectrum efficiency versus altitude of ABSs}\label{fig_6}
	\end{figure}
	Next, Fig. \ref{fig_5} plots the coverage probability of a typical UE versus the SINR threshold when it is associated with an ABS both in the NOMA and OMA schemes. The results show that  with the increase of the SINR threshold, the coverage probability decreases both for the NOMA and OMA schemes. It is also observed that the OMA scheme outperforms  some particular NOMA scheme when the SINR threshold is large. This again demonstrates  the importance of using power allocation. Besides,  the coverage probability is always zero in the case of inappropriate power allocation factors given the target SINR threshold of the typical UE and the fixed UE, which  coincides with \textbf{Remark 4}.
	
	We then study the impact of network density on the  coverage probability versus the density of ABSs when the typical UE is associated with an ABS. Fig. \ref{fig_7} shows the coverage behaviour when the targeted SINR threshold is 0 dB and 2 dB with $h=200$ m and $h=50$ m, respectively. The density of TBSs is fixed as $10^{-5}/m^2$. It is observed that  when $h=200$ m, with the increase of the density of the ABSs, the coverage probability of the typical UE first increases and then decreases. This is because as we increase the density of the ABS, the  probability that the ABSs can be selected as the serving BS will be increased. When the ABSs density is further increased, the interference from the non-serving ABSs will  become larger, causing the decrease of the coverage probability. On the other hand, when $h=50$ m, the coverage probability of the typical UE has an ascending trend. This is due to the fact that at low altitude interference is dominated by the NLoS signals. The results also show that under these settings, the NOMA scheme outperforms the OMA scheme, which again verifies the effectiveness of NOMA.

	Finally, we  verify the rate performance of the proposed UAV-aided HetNets. Fig. \ref{fig_6} shows the spectrum efficiency versus the altitude of ABSs.  One  can observe that due to the deployment of multiple antennas, mmWave can achieve a high spectrum efficiency regardless the change of ABSs altitude. In the NOMA scheme, the spectrum efficiency  decreases with the increase of the altitude of ABSs with the perfect SIC, which indicates that the interference increases at a larger rate than  the increase in the desired signals. Combined with the results in Fig. \ref{fig_1}, we notice that with the increase of  ABSs altitude, even the association probability connected to an ABS increases,  the spectrum efficiency will decrease. This phenomenon reveals that the interference from the LoS ABSs should be carefully considered when we design the networks. We also notice the spectrum efficiency is affected to a large extent by ipSIC coefficient. On the other hand, in the OMA scheme, the  spectrum efficiency will first increase and then decrease,  due to at low altitude, the probability of LoS A2G transmissions is small.  These results also verify the effectiveness of our proposed UAV-aided HetNets.
	
	\section{Conclusions}
	In this paper, a novel UAV-aided HetNets model consisting the NOMA-based ABSs and mmWave TBSs has been proposed and a flexible NOMA based UE association policy  has been investigated. By using the tools from stochastic geometry, we derived the analytical expressions for the distance distributions  given that the typical UE is associated with a TBS, a NLoS ABS or a LoS ABS. Additionally, new  analytical expressions for association  probability,  coverage probability and spectrum efficiency  have also been derived for  characterizing the performance of UAV-aided HetNets under  the realistic A2G/G2G channels.  Moreover, we provided the approximated expressions for the coverage probability and spectrum efficiency  to simplify the analytical results. Finally, in the numerical results, we provided insights  for the design of the  HetNets  i.e., the UEs tend to associate with the TBSs when the ABSs are deployed at low altitudes in a dense environment.We studied the impact of the ABSs altitude, network density, NOMA power allocation, SINR threshold and ipSIC coefficient in terms of the whole network performance. Analytical and  simulated results demonstrated that terrestrial mmWave small cells can offer high capacity and the ABSs with NOMA are capable of achieving superior performance compared with the  ABSs with OMA.

	\begin{appendices}
		\renewcommand{\thesectiondis}[2]{\Alph{section}:}
		\section{Proof of Lemma \ref{association_L}}\label{Appendix_A}
		When the typical UE is associated to a LoS ABS, and there does not exist a TBS in $\mathcal{B}(0,R_B)$,
		\begin{align}
			A_{L,1}&=(1-Q_T)\mathbb{P}(A_{L,1}^{N})\nonumber\\
			&=(1-Q_T)\mathbb{P}(\eta_LR_{L,0}^{-\alpha_L}>\eta_NR_{N,0}^{-\alpha_N})\nonumber\\
			&=(1-Q_T)\Bigg[\int_{h}^{l_{L,h}}f_{R_{L,0}}(r)dr+\int_{l_{L,h}}^{\infty}  \exp\Bigg(-2\pi\lambda_A\nonumber\\
			&\quad\times\int_{0}^{\sqrt{\tau^2_{N|L}(r)-h^2}}xP_N(x)dx\Bigg) f_{R_{L,0}}(r) dr\Bigg],
		\end{align}
		where $A_{L,1}^N$ denotes the probability that the power received from the closest  LoS ABS is stronger than that received from the closest  NLoS ABS in the above scenario.
		
		When there exists a TBS in $\mathcal{B}(0,R_B)$, the corresponding association probability is given by
		\begin{equation}
			A_{L,2}=Q_TQ_L\mathbb{P}(A_{L,2}^{N})\mathbb{P}(A_{L,2}^{T}),
		\end{equation}
		where $A_{L,2}^N$ and  $A_{L,2}^T$ denote the probability that the power received from the closest  LoS ABS is stronger than that received from the closest NLoS ABS  and TBS in the above scenario, respectively. Then we have
		\begin{align}
			\mathbb{P}(A_{L,2}^{T})&=\frac{\mathbb{P}\left(R_{T,0}>\left(\frac{\eta_T}{\eta_L}\right)^{\frac{1}{\alpha_T}}R_{L,0}^{\frac{\alpha_L}{\alpha_T}}|R_{T,0}\leq R_B\right)}{Q_LQ_T}\nonumber\\
			&\overset{(a)}{=}\frac{\int_{h}^{l_{L,T}}[\mathbb{P}(R_{T,0}>\tau_{T|L}(r))-(1-Q_T)]f_{R_{L,0}}(r)dr}{Q_LQ_T}\nonumber\\
			&=\frac{\int_{h}^{l_{L,T}}  \left[\exp\left(-\pi\lambda_T \tau^2_{T|L}(r) \right)-(1-Q_T)\right]{f_{R_{L,0}}(r)}}{Q_LQ_T},
		\end{align}
		where $(a)$ is due to the fact that there exists a TBS in $\mathcal{B}(0,R_B)$.
		
		Similarly, we can derive the expression of $\mathbb{P}(A_{L,2}^{N})$ shown in (\ref{A_L2_N}).

		
		
		\section{Proof of Lemma \ref{laplace_T}}\label{Appendix_C}
		When the typical UE is associated with a TBS, the interference from the TBSs can be expressed as $I_T=\sum_{x_{T,i}\in\Phi_T\cap\mathcal{B}(0,R_B)\backslash x_{T,0}} G_i P_TC_TR_{T,i}^{-\alpha_T}H_{T,i} $, then the expressions in (\ref{app_B_1}) which is shown at the top of next page can be obtained,
		\begin{figure*}[!t]
			\begin{equation}\label{app_B_1}
				\begin{split}
					L_{I_T}(s)&=\mathbb{E}_{I_T}\left[e^{-sI_T}\right]\\
					&\overset{(a)}{=}\mathbb{E}_{\Phi_T}\left[\prod_{x_{T,i}\in\Phi_T\cap\mathcal{B}(0,R_B)\backslash x_{T,0}}\mathbb{E}_{H_{T,i}}\left[e^{-sG_i P_TC_TR_{T,i}^{-\alpha_T}H_{T,i}}\right]\right]\\
					&\overset{(b)}{=}\exp\left(-2\pi\lambda_T\sum_{j\in\{M,m\}}p_j\int_{R_{T,0}}^{R_B}\left(1-\mathbb{E}_{H_{T,i}}\left[e^{-sG_j P_TC_Tr^{-\alpha_T}H_{T,i}}\right]\right)rdr\right)\\
					&\overset{(c)}{=}\exp\left(-2\pi\lambda_T\sum_{j\in\{M,m\}}p_j\int_{R_{T,0}}^{R_B}\left(1-\left(1+\frac{sP_TC_TG_j}{m_Tr^{\alpha_T}}\right)^{-m_T}\right)rdr\right)\\
					&\overset{(d)}{=}\exp\left(-2\pi\lambda_T\sum_{j\in\{M,m\}}\sum_{i=1}^{m_T}p_j\binom{m_T}{i}\left(\frac{sP_TC_TG_j}{m_T}\right)^i\int_{R_{T,0}}^{R_B}\frac{r^{-\alpha_Ti+1}}{\left(1+\frac{sP_TC_TG_j}{m_Tr^{\alpha_T}}\right)^{m_T}}dr\right)\\
					&\overset{(e)}{=}\exp\left(\frac{-2\pi\lambda_T}{\alpha_T}\sum_{j\in\{M,m\}}\sum_{i=1}^{m_T}p_j\binom{m_T}{i}\left(\frac{sP_TC_TG_j}{m_T}\right)^{\frac{2}{\alpha_T}}(-1)^{\frac{2}{\alpha_T}-i+1}\int_{t_{j,l}}^{t_{j,u}}\frac{t_j^{i-\frac{2}{\alpha_T}-1}}{\left(1-t_j\right)^{m_T}}dt_j\right).
				\end{split}
			\end{equation}
			\hrulefill
		\end{figure*}
		where $(a)$ is due to the properties of exponential terms, $(b)$ is obtained by using probability generating functional (PGFL) \cite{new_3},  $(c)$ is  obtained by computing the moment generating function of a Gamma random variable, $(d)$ follows from the binomial theorem, and $(e)$ is obtained by adopting $t_j=-\frac{sP_TC_TG_j}{m_Tr^{\alpha_T}}$.  Finally, from \cite{table}, $L_{I_T}(s)$ can be expressed as the form in (\ref{laplace_TBS}).
		
		\section{Proof of Lemma \ref{coverage_T}}\label{Appendix_D}
		The coverage probability of a typical UE when it is associated with a TBS can be expressed as (\ref{app_C_1}) which is shown at the top of next page,
		\begin{figure*}[htb]
			\begin{equation}\label{app_C_1}
				\begin{split}
					P_T^C&=\int_{0}^{R_B}\mathbb{P}\left(H_{T,0}\geq r^{\alpha_T}\frac{\nu_T}{G_MP_TC_T}(I_T+  \sigma^2_T)\right)f_{\hat{R}_{T,0}}(r)dr\\
					&\overset{(a)}{=}\int_{0}^{R_B}\left(\mathbb{E}_{I_T}\left[\exp\left(-m_Tr^{\alpha_T}\varepsilon_TI_T\right)\cdot\sum_{k=0}^{m_T-1}\frac{\left(m_Tr^{\alpha_T}\varepsilon_T(I_T+  \sigma^2_T)\right)^k}{k!}\right]\exp\left(-m_Tr^{\alpha_T}\varepsilon_T  \sigma^2_T\right)\right)f_{\hat{R}_{T,0}}(r)dr\\
					&\overset{(b)}{=}\int_{0}^{R_B}\Bigg(\sum_{k=0}^{m_T-1}\sum_{p=0}^{k}\binom{k}{p}\frac{r^{\alpha_Tk}(m_T\varepsilon_T)^k}{k!}\exp\left(-m_Tr^{\alpha_T}\varepsilon_T  \sigma^2_T\right)\left(\sigma^2_T\right)^p\mathbb{E}_{I_T}\left[\exp\left(-m_Tr^{\alpha_T}\varepsilon_TI_T\right)(I_T)^{k-p}\right]\Bigg)f_{\hat{R}_{T,0}}(r)dr\\
					&\overset{(c)}{=}\int_{0}^{R_B}\Bigg(\sum_{k=0}^{m_T-1}\sum_{p=0}^{k}\binom{k}{p}\frac{r^{\alpha_Tk}(m_T\varepsilon_T)^k}{k!}\exp\left(-m_Tr^{\alpha_T}\varepsilon_T  \sigma^2_T\right)\left(\sigma^2_T\right)^p\left[(-1)^{k-p}\frac{\partial^{k-p}}{\partial s^{k-p}}L_{I_T}(s)\right]_{s=m_T\varepsilon_Tr^{\alpha_T}}\Bigg)f_{\hat{R}_{T,0}}(r)dr.
				\end{split}
			\end{equation}
			\hrulefill
		\end{figure*}	
		where $(a)$ is obtained from the CCDF of Gamma distribution, i.e. $\bar{F}_G(g)=\frac{\Gamma_u(m_T,m_Tg)}{\Gamma(m_T)}=\exp(-m_Tg)\sum_{k=0}^{m_T-1}\frac{(m_Tg)^k}{k!}$, $(b)$ follows from the binomial theorem, and $(c)$ is from the fact that $\mathbb{E}_{I_T}[\exp(-sI_T)I_T^i]=(-1)^i\frac{\partial^i}{\partial s^i}L_{I_T}(s)$.
		
		Since the expression (\ref{app_C_1}) involves  higher order derivatives of the Laplace transform which needs the aid of Faa Di Bruno's formula\cite{Faa_Di_Bruno}. To reduce the high complexity, we adopt the Alzer's Lemma to give an approximation  of a Gamma random variable\cite{gamma_bound}, which is given by
		\begin{equation}
			\mathbb{P}(g<\tau)\approx\left[1-e^{-b_T\tau}\right]^{m_T},
		\end{equation}
		where $b_T=m_T(m_T!)^{-\frac{1}{m_T}}$, then we have
		\begin{align}
			&P_T^C=\int_{0}^{R_B}\left(1-\mathbb{P}\left(H_{T,0}< r^{\alpha_T}\varepsilon_T\left(I_T+\sigma^2_T\right)\right)\right)f_{\hat{R}_{T,0}}(r)dr\nonumber\\
			&{\approx}\int_{0}^{R_B}\left(1-\mathbb{E}_{I_T}\left[\left(1-e^{-b_Tr^{\alpha_T}\varepsilon_T\left(I_T+ \sigma^2_T\right)}\right)^{m_T}\right]\right)f_{\hat{R}_{T,0}}(r)dr\nonumber\\
			&\quad\overset{(d)}{=}\int_{0}^{R_B}\Bigg(\sum_{k=1}^{m_T}(-1)^{k+1}\binom{m_T}{k}e^{-kb_Tr^{\alpha_T}\varepsilon_T\sigma^2_T}\nonumber\\
			&\quad\quad\quad\times\mathbb{E}_{I_T}\left[e^{-kb_Tr^{\alpha_T}\varepsilon_TI_T}\right]\Bigg)f_{\hat{R}_{T,0}}(r)dr\nonumber\\
			&\quad\overset{(e)}{=}\int_{0}^{R_B}\Bigg(\sum_{k=1}^{m_T}(-1)^{k+1}\binom{m_T}{k}e^{-kb_Tr^{\alpha_T}\varepsilon_T\sigma^2_T}\nonumber\\
			&\quad\quad\quad\times L_{I_T}(kb_Tr^{\alpha_T}\varepsilon_T)\Bigg)f_{\hat{R}_{T,0}}(r)dr,
		\end{align}
		where  $(d)$ follows from the binomial theorem, and $(e)$ is obtained using the Laplace transform of interference.
	\end{appendices}
	
	\vspace{-1em}
	\bibliography{references}


\end{document}